\documentclass[a4paper]{article}
\usepackage{amsmath}
\usepackage[english]{babel}
\usepackage{latexsym}
\usepackage{amssymb}
\usepackage{amscd}
\usepackage{amsgen,amstext,amsbsy,amsopn}
\usepackage{math rsfs}

\usepackage{amsthm,epsfig,graphicx,graphics}
\usepackage[latin1]{inputenc}
\usepackage{xspace}
\usepackage{amsxtra}

\usepackage{color}

\newcommand{\version}{March 10th, 2011}
\renewcommand{\theequation}{\thesection.\arabic{equation}}
\numberwithin{equation}{section}
\newcommand{\bdm}{\begin{displaymath}}
\newcommand{\edm}{\end{displaymath}}
\newcommand{\bdn}{\begin{eqnarray}}
\newcommand{\edn}{\end{eqnarray}}
\newcommand{\bay}{\begin{array}{c}}
\newcommand{\eay}{\end{array}}
\newcommand{\ben}{\begin{enumerate}}
\newcommand{\een}{\end{enumerate}}
\newcommand{\beq}{\begin{equation}}
\newcommand{\eeq}{\end{equation}}
\newcommand{\bml}[1]{\begin{multline} #1 \end{multline}}
\newcommand{\bmln}[1]{\begin{multline*} #1 \end{multline*}}

\newcommand{\lf}{\left}
\newcommand{\ri}{\right}

\newcommand{\RR}{\mathbb{R}^2}

\newcommand{\rv}{\vec{r}}
\newcommand{\rvi}{\vec{r}_i}

\newcommand{\avj}{\vec{a}_j}
\newcommand{\avi}{\vec{a}_i}

\newcommand{\diff}{\mathrm{d}}
\newcommand{\eps}{\varepsilon}

\newcommand{\chin}{\chi_{\mathrm{in}}}
\newcommand{\chout}{\chi_{\mathrm{out}}}
\newcommand{\xiin}{\xi_{\mathrm{in}}}
\newcommand{\xiout}{\xi_{\mathrm{out}}}

\newcommand{\Ofirst}{\Omega_{\mathrm{c_1}}}
\newcommand{\Osec}{\Omega_{\mathrm{c_2}}}
\newcommand{\Othird}{\Omega_{\mathrm{c_3}}}

\newcommand{\ba}{\mathcal{B}}

\newcommand{\gpf}{\mathcal{E}^{\mathrm{GP}}}
\newcommand{\gpe}{E^{\mathrm{GP}}}
\newcommand{\gpm}{\Psi^{\mathrm{GP}}}
\newcommand{\chem}{\mu^{\mathrm{GP}}}

\newcommand{\gpdom}{\mathscr{D}^{\mathrm{GP}}}
\newcommand{\gpdomt}{\tilde{\mathscr{D}}^{\mathrm{GP}}}

\newcommand{\hgpf}{\hat{\mathcal{E}}^{\mathrm{GP}}}
\newcommand{\hgpe}{\hat{E}^{\mathrm{GP}}}
\newcommand{\hchem}{\hat{\mu}^{\mathrm{GP}}}
\newcommand{\hgpm}{g}
\newcommand{\hgpd}{\hat{\rho}}
\newcommand{\hgpdom}{\hat{\mathscr{D}}^{\mathrm{GP}}}

\newcommand{\gvf}{\mathcal{E}^{\mathrm{gv}}_{\omega}}
\newcommand{\gvfo}{\mathcal{E}^{\mathrm{gv}}_{\omega_0}}
\newcommand{\gvfopt}{\mathcal{E}^{\mathrm{gv}}_{\oopt}}
\newcommand{\gve}{E^{\mathrm{gv}}_{\omega}}

\newcommand{\giante}{E^{\mathrm{gv}}}
\newcommand{\gveopt}{E^{\mathrm{gv}}_{\oopt}}
\newcommand{\gveo}{E^{\mathrm{gv}}_{\omega_0}}
\newcommand{\gvmopt}{g_{\mathrm{opt}}}

\newcommand{\hgvf}{\tilde{\mathcal{E}}^{\mathrm{gv}}_{\omega}}
\newcommand{\hgve}{\tilde{E}^{\mathrm{gv}}_{\omega}}

\newcommand{\hgveo}{\tilde{E}^{\mathrm{gv}}_{\omega_0}}
\newcommand{\hgvm}{g_{\omega}}
\newcommand{\hgvmo}{g_{\omega_0}}

\newcommand{\hgvchem}{\tilde{\mu}^{\mathrm{gv}}_{\omega}}
\newcommand{\hgvchemo}{\tilde{\mu}^{\mathrm{gv}}_{\omega_0}}

\newcommand{\gpfi}{\mathcal{E}^{(i)}}

\newcommand{\Eg}{\mathcal{E}}
\newcommand{\Fg}{\mathcal{F}}

\newcommand{\dg}{{\rm deg}}
\newcommand{\curl}{{\rm curl}}

\newcommand{\set}{\mathcal{S}}

\newcommand{\tff}{\mathcal{E}^{\mathrm{TF}}}
\newcommand{\tfd}{\mathcal{A}^{\mathrm{TF}}}
\newcommand{\tfe}{E^{\mathrm{TF}}}
\newcommand{\tfm}{\rho^{\mathrm{TF}}}
\newcommand{\tfdom}{\mathscr{D}^{\mathrm{TF}}}
\newcommand{\rtf}{R_{\mathrm{h}}}
\newcommand{\rd}{\bar R}
\newcommand{\rb}{R_{>}}
\newcommand{\rt}{R_{<}}
\newcommand{\rtilde}{\tilde{R}}

\newcommand{\rmax}{R_{\mathrm{m}}}
\newcommand{\rmaxn}{R_{*}}
\newcommand{\rmaxgv}{\tilde{R}_{\mathrm{m}}}

\newcommand{\tfchem}{\mu^{\mathrm{TF}}}

\newcommand{\htff}{\tilde{\mathcal{E}}^{\mathrm{TF}}}
\newcommand{\htfe}{\tilde{E}^{\mathrm{TF}}}

\newcommand{\At}{\mathcal{A}_{\mathrm{bulk}}}
\newcommand{\ann}{\mathcal{A}}
\newcommand{\supp}{\mathrm{supp}}
\newcommand{\tfsupp}{\mathrm{supp}\lf(\tfm\ri)}

\newcommand{\optphtf}{\omega^{\mathrm{TF}}}
\newcommand{\oopt}{\omega_{{\rm opt}}}

\newcommand{\gaintf}{H^{\mathrm{TF}}}

\newcommand{\gain}{H}
\newcommand{\costtf}{F^{\mathrm{TF}}}

\newcommand{\const}{C}

\newcommand{\trial}{\Psi_{\mathrm{trial}}}

\newcommand{\latt}{\mathcal{L}}
\newcommand{\hlatt}{\hat{\mathcal{L}}}
\newcommand{\spac}{\ell}
\newcommand{\cell}{\mathcal{Q}}
\newcommand{\celli}{\mathcal{Q}^i_{}}

\newcommand{\hspac}{\hat{\ell}}
\newcommand{\hcell}{\hat{\mathcal{Q}}}
\newcommand{\hcelli}{\hat{\mathcal{Q}}^i_{}}

\newcommand{\magnp}{\vec{A}}
\newcommand{\rmagnp}{\vec{B}_{\omega}}
\newcommand{\ftrial}{f_{\mathrm{trial}}}

\newcommand{\half}{\hbox{$\frac12$}}
\newcommand{\thalf}{\hbox{$\frac32$}}
\newcommand{\fout}{F_{\mathrm{out}}}
\newcommand{\fin}{F_{\mathrm{in}}}

\newcommand{\Z}{\mathbb{Z}}

\newcommand{\N}{\mathbb{N}}

\newcommand{\D}{\mathcal{D}}
\newcommand{\F}{\mathcal{F}}
\newcommand{\E}{\mathcal{E}}
\newcommand{\A}{\mathcal{A}}
\newcommand{\B}{\mathcal{B}}

\newcommand{\Dt}{\tilde{\mathcal{D}}}

\newcommand{\Q}{\mathcal{Q}}
\newcommand{\OO}{\mathcal{O}}

\newcommand{\al}{\alpha}
\newcommand{\alt}{\tilde{\alpha}}
\newcommand{\ep}{\varepsilon}

\newcommand{\Om}{\Omega}
\newcommand{\om}{\omega}

\newcommand{\dd}{\partial}

\newtheorem{teo}{Theorem}[section]
\newtheorem{lem}{Lemma}[section]
\newtheorem{pro}{Proposition}[section]

\newtheorem{defi}{Definition}[section]

\newcounter{remark}[section]
\newenvironment{rem}{\stepcounter{remark} \vspace{0,1cm} \noindent \textit{Remark \thesection.\theremark}\,}{\vspace{0,2cm}}

\begin{document}

\markboth{\scriptsize{Critical Speeds in GP Theory -- CPRY -- \version}}{\scriptsize{Critical Speeds in the GP Theory -- CPRY  -- \version}}

\title{Critical Rotational Speeds in the Gross-Pitaevskii Theory on a Disc with 
Dirichlet Boundary Conditions}

\author{M. Correggi${}^{a}$\footnote{Supported by a grant ({\it assegno di ricerca}) of {\it Istituto Nazionale di Alta Matematica ``F. Severi''}.} , F. Pinsker${}^{b}$, N. Rougerie${}^{c}$, J. Yngvason${}^{d,e}$	\\
	\mbox{}	\\
	\normalsize\it ${}^{a}$ CIRM, Fondazione Bruno Kessler, Via Sommarive 14, 38123 Trento, Italy.	\\
	\normalsize\it ${}^{b}$ DAMTP, University of Cambridge, Wilbertforce Road, Cambridge CB3 0WA, United Kingdom.\\
	\normalsize\it ${}^{c}$ CNRS et Universit\'e de Cergy-Pontoise, D\'{e}partement de Math\'{e}matiques, CNRS-UMR 
8088, \\ \normalsize\it Site de Saint Martin, 2 avenue Adolphe Chauvin, 95302 Cergy-Pontoise Cedex. \\
	\normalsize\it ${}^{d}$ Fakult\"at f\"ur Physik, Universit{\"a}t Wien, Boltzmanngasse 5, 1090 Vienna, Austria.	\\
	\normalsize\it ${}^{e}$ Erwin Schr{\"o}dinger Institute for Mathematical Physics, Boltzmanngasse 9, 1090 Vienna, Austria.}	
\date{\version}

\maketitle

\begin{abstract} We study the two-dimensional Gross-Pitaevskii theory of a rotating Bose gas in a disc-shaped trap with Dirichlet boundary conditions, generalizing and extending previous results that were obtained under Neumann boundary conditions. The focus is on the energy asymptotics, vorticity and qualitative properties of the minimizers in the parameter range $|\log\eps|\ll\Omega\lesssim\eps^{-2}|\log\eps|^{-1}$ where $\Omega$ is the rotational velocity and the coupling parameter is written as $\eps^{-2}$ with $\eps\ll 1$. Three critical speeds can be identified. At $\Omega=\Ofirst\sim|\log\eps|$ vortices start to appear and for 
$|\log\eps|\ll \Omega< \Osec\sim \eps^{-1}$ the vorticity is uniformly distributed over the disc. For $\Omega\geq\Osec$ the centrifugal forces create a hole around the center with strongly depleted density. For $\Omega\ll \eps^{-2}|\log\eps|^{-1}$ vorticity is still uniformly distributed in an annulus containing the bulk of the density,  but at $\Omega=\Othird\sim\eps^{-2}|\log\eps|^{-1}$ there is a transition to a giant vortex state where the vorticity disappears from the bulk. The energy is then well approximated by a trial function that is an eigenfunction of angular momentum but one of our results is that the true minimizers break rotational symmetry in the whole parameter range, including the giant vortex phase.

	\vspace{0,2cm}

	\vspace{0,2cm}
	
\end{abstract}

\tableofcontents

\section{Introduction and Main Results}

The Gross-Pitaevskii (GP) theory is the most commonly used model to describe the behavior of rotating superfluids. Since the nucleation of quantized vortices is a signature of the superfluid behavior it is of great interest to understand that phenomenon in the framework of the GP theory. A fascinating example of superfluid is provided by a cold Bose gas forming a Bose-Einstein condensate (BEC). The possibility to nucleate quantized vortices in a rotating BEC has triggered a lot of interest in the last decade, both experimental and theoretical (see the reviews \cite{Co,Fe1} and the monograph \cite{A} for further references).

Bose-Einstein condensates are trapped systems: A magneto-optical confinement is imposed on the atoms. When rotating such a system, the strength of the confinement can lead to two different behaviors. If the trapping potential increases quadratically with the distance from the rotation axis (`harmonic' trap), there exists a limiting angular velocity that one can impose to the gas. Any larger velocity would result in a centrifugal force stronger than the trapping force. The atoms would then be driven out of the trap. By contrast, a stronger confinement (`anharmonic' trap)  allows in principle an arbitrary angular velocity. In this paper we focus on the two-dimensional GP theory for a BEC with anharmonic confinement. 

Theoretical and numerical arguments have been proposed in the physics literature (see, e.g., \cite{FJS,FB,KB}) in favor of the existence of three critical speeds at which important phase transitions are expected to happen: 
\begin{itemize}
\item If the velocity $\Om$ is smaller than the first critical velocity $\Ofirst$, then there are no vortices in the condensate (`vortex-free state');
\item If $ \Omega $ is between $ \Ofirst $ and $ \Osec$, there is a hexagonal lattice of singly quantized vortices (`vortex-lattice state');
\item When $\Om$ is taken larger than $\Osec$, the centrifugal force becomes so important that it dips a hole in the center in the condensate. The annulus in which the mass is concentrated still supports a vortex lattice however (`vortex-lattice-plus-hole state'), until $\Om$ crosses the third threshold $\Othird$;
\item If $\Om $ is larger than $ \Othird$, all vortices retreat in the central low density hole, resulting in a `giant vortex' state. The central hole acts as a multiply quantized vortex with a large phase circulation.
\end{itemize}
In \cite{CDY1,CY,CRY,R} we have studied these phase transitions using as model case a BEC in a `flat' trap, that is a constant potential with hard walls. This is the `most anharmonic' confinement one can imagine and serves as an approximation for potentials used in experiments. Mathematically, it has the advantage that the rescaling of spatial variables as $\eps\to0$ and/or $\Omega\to\infty$ is avoided. The GP energy functional in the non-inertial rotating frame is defined as
\beq
    \label{GPf}
    \gpf[\Psi] : = \int_{\ba} \diff \vec{r} \: \left\{ |\nabla \Psi|^2 - 2 \Psi^* \vec{\Omega}  \cdot \vec{L} \Psi + \eps^{-2} |\Psi|^4 \right\}
\eeq
where we have denoted the physical angular velocity by $2\vec{\Omega}$, $\vec L=-i \rv \wedge \nabla $ is the angular momentum operator and $ \ba $ the unit two-dimensional disc. We have written the coupling constant as $\ep^{-2}$. The subsequent analysis (as well as the papers \cite{CDY1,CY,CRY,R}) is concerned about the `Thomas-Fermi' (or strongly interacting) limit where $\ep \rightarrow 0$.\\
The  simplest way to define the ground state of the system is to minimize the energy functional (\ref{GPf}) under the mass constraint
\[
\int_{\ba} \diff \rv \: \left| \Psi \right|^2 = 1
\] 
with no further conditions. This is the approach that has been considered in the previous papers \cite{CDY1,CY,CRY,R}, leading to Neumann boundary conditions on $\dd \ba$. We will refer to this situation as the `flat Neumann problem' in the sequel. 

There are, however, both physical and mathematical reasons for considering also the corresponding problem with a Dirichlet boundary condition, i.e., requiring the wave function to vanish on the boundary of the unit disc. Physically, this corresponds to a hard, repelling wall which is usually a closer approximation to real experimental situations than a `sticky' wall modeled by a Neumann boundary condition. The Dirichlet boundary condition can be formally implemented by replacing the flat trap with a smooth confining potential of the form $r^s$ and taking\footnote{This limit has to be taken with care, however, because it can not be interchanged with the asymptotic limit $\eps\Omega\to\infty$ we shall consider. This point will be discussed further in \cite{CPRY}.} $s\to\infty$. 

Mathematically, the new boundary condition is responsible for some new aspects requiring several modifications of the proofs. For one thing,  the density profile is no longer a monotonously increasing function of the radial variable and the position of the density maximum has to be precisely estimated. Furthermore,  energy estimates have to be refined to take the boundary effect into account, and a boundary estimate for the GP minimizer, that was an important ingredient in the proof of the giant vortex transition in \cite{CRY}, has to be replaced by a different approach. 

In addition to these adaptations to the new situation the present paper contains also substantial improvements of results proved previously in the Neumann case. These concern in particular the uniform distribution of vorticity in the bulk (Theorem \ref{uniform distribution}) and the rotational symmetry breaking (Theorem \ref{symmetry breaking}). Besides, the error term in our energy estimate in Theorem \ref{giant vortex teo} below is much smaller than the corresponding term in \cite[Theorem 1.2]{CRY}. This last improvement is due to the new method for estimating a potential function that we use to avoid the boundary estimate.

\vspace{0,3cm}

From now on the minimization of (\ref{GPf}) is considered on the domain   
\beq
	\label{minimization dom}
	\gpdom : = \lf\{ \Psi \in H^1_0(\ba) : \: \lf\| \Psi \ri\|_2 = 1 \ri\},
\eeq
where $H^1_0(\ba)$ is the Sobolev space of complex valued functions $\Psi$ on $\ba$ with $\int_\ba(|\Psi|^2+|\nabla\Psi|^2)<\infty$ and $\Psi(\rv)=0$ on $\partial\ba$. The ground state energy is thus defined as
\beq
	\label{GPe} 
	\gpe : = \inf_{\Psi \in \gpdom} \gpf[\Psi],
\eeq
and any corresponding minimizer is denoted by $ \gpm $. This case will be referred to as the `flat Dirichlet problem'. In the following we will often use a different form of the GP functional which can be obtained by introducing a vector potential, i.e.,
\beq
	\label{GPf magn}
	\gpf[\Psi] = \int_{\ba} \diff \rv \: \lf\{ \lf| \lf( \nabla - i \magnp \ri) \Psi \ri|^2 - \Omega^2 r^2 |\Psi|^2 + \eps^{-2} |\Psi|^4 \ri\},
\eeq
where 
\beq
	\magnp : = \vec\Omega\wedge \rv= \Omega r \vec{e}_{\vartheta}.
\eeq
Here $(r,\vartheta)$ are two-dimensional polar coordinates and $\vec{e}_{\vartheta}$ a unit vector in the angular direction.

\vspace{0,3cm}

The GP minimizer is in general not unique because vortices can break the rotational symmetry (see Section \ref{sec symm break}) but any minimizer satisfies in the open ball the variational equation (GP equation) 
\beq
	\label{GP variational}
	- \Delta \gpm - 2 \vec\Omega\cdot \vec L\, \gpm + 2 \eps^{-2} \lf| \gpm \ri|^2 \gpm = \chem \gpm,
\eeq
with additional Dirichlet conditions at the boundary, i.e., 
\beq
	\label{Dirichlet bc}
	\gpm(\rv) = 0 \:\:\: \mbox{on} \:\:\: \partial \ba.
\eeq
The chemical potential in \eqref{GP variational} is given by the normalization condition on $ \gpm $, i.e.,
\beq
	\label{chem}
	\chem : = \gpe + \frac{1}{\eps^2} \int_{\ba} \diff \rv \: \lf| \gpm \ri|^4.
\eeq

For such a model, variational arguments have been provided in \cite{FB} to support the following conjectures about the three critical speeds:
\begin{eqnarray}
\Ofirst &\propto& |\log \ep |,\\
\Osec &\propto& \ep ^{-1},\\
\Othird &\propto& \ep ^{-2} |\log \ep| ^{-1}.
\end{eqnarray} 
As for the behavior of the condensate close to $\Ofirst$, the centrifugal force is not strong enough for the specificity of the anharmonic confinement to be of importance. A consequence is that the analysis developed in \cite{IM,IM2} (see also \cite{AJR} for recent developments) for harmonic traps applies and leads to the rigorous estimate
\beq \label{estimOfirst}
 \Ofirst = |\log \ep| (1+o(1))
\eeq
when $\ep \rightarrow 0$. In this paper we aim at providing estimates of $\Osec$ and $\Othird$ and thus will assume that
\[
\Om \gg |\log \ep|,
\] 
i.e., we consider angular velocities strictly above $\Ofirst$. The situation is then very different from that in a harmonic trap because of the onset of strong centrifugal forces when $\Om$ approaches $\Osec$.

\vspace{0,3cm}

Our main results can be summarized as follows. We show that if $\Om \leq 2 (\sqrt{\pi}\ep) ^{-1}$, the condensate is disc-shaped, while for $\Om > 2 (\sqrt{\pi}\ep) ^{-1}$ the matter density is confined in an annulus along the boundary of $\ba$. In addition we prove that if
\[
|\log \ep| \ll \Om \ll \frac{1}{\ep ^2 |\log \ep|},
\]
there is a uniform distribution of vorticity in the bulk of the condensate. Although our estimates are not precise enough to show that there is a hexagonal lattice of vortices, these results support the qualitative picture provided in \cite{FB}. We deduce that when $\ep \rightarrow 0$
\beq \label{estimOsec}
\Osec = \frac{2}{\sqrt{\pi} \ep } (1+o(1)).
\eeq
We refer to Section \ref{sec unif distr} for the detailed statements of these results.

In Section \ref{sec giant vortex} we present our results about the third critical speed. We show that if $\Om = \Om_0 \ep ^{-2} |\log \ep| ^{-1}$ with $\Om_0 > 2(3 \pi) ^{-1}$, then there are no vortices in the bulk of the condensate. This provides an upper bound on the third critical speed
\beq \label{estimOthird}
\Othird \leq \frac{2}{3\pi \ep ^2 |\log \ep| } (1+o(1)).
\eeq
It should be noted right away that we do believe that this upper bound is optimal. This has been proved in \cite{R} in the flat Neumann case and the adaptation of the adequate tools to the flat Dirichlet case is possible but beyond the scope of this paper. We hope to come back to the regime $\Om \propto \Othird$ in the future.

We also remark that the estimates we obtain for the three critical speeds in the limit $\ep \rightarrow 0$ are the same in the flat Neumann and Dirichlet settings. In the cases of the first and second critical velocities this is plausible because the features that mark the onset of the transition (the first vortices and the appearance of the 'hole' respectively) occur far from the boundary of the trap. The independence of the third critical velocity of boundary conditions is less obvious but the main reason is that the maximum of the density is to leading order the same for both boundary conditions. 
\newline
\indent In the regime $\Om > \Othird$ a very natural question occurs about the distribution of vorticity in the central hole of low matter density: Is the phase of the condensate created by a single multiply quantized vortex at the center of the trap? We show that this is not the case in Section \ref{sec symm break} and, as a consequence, the rotational symmetry is always broken at the level of the ground state, even when $\Om > \Othird$. 

\vspace{0,3cm}

Before stating our results more precisely, we want to make a comparison with the 2D Ginzburg-Landau (GL) theory for superconductors in applied magnetic fields (see \cite{BBH,FH,SS2} for a mathematical presentation). The analogies between GP and GL theories have often been pointed out in the literature, with the external magnetic field playing in GL theory the role of the angular velocity in GP theory. We stress that our results in fact enlighten significant differences between the two theories. Whereas the first critical speed in GP theory can be seen as the equivalent of the first critical field in GL theory, the second and third critical speeds have little to do with the second and third critical fields of the GL theory. The difference can be seen both in the order of magnitudes of these quantities as functions of $\ep$ (which for a superconductor is the inverse of the GL parameter) and in the qualitative properties of the states appearing in the theories. In GP theory there is no equivalent of the normal state and there is no vortex-lattice-plus-hole state in GL theory. The giant vortex state of GP theory could be compared to the surface superconductivity state in GL theory, but the physics governing the onset of these two phases is quite different. The main reason for this different behavior is the combined influence of the centrifugal force and mass constraint in GP theory, two features that have no equivalent in GL theory.

\vspace{0,3cm}

We will now state our results rigorously. The core analysis that we present below is an adaptation of the techniques developed in \cite{CDY1,CY,CRY} for the Neumann case, but the Dirichlet condition leads to important novel aspects that we discuss in the sequel.

\subsection{The Regime $ |\log\eps| \ll \Omega \ll \eps^{-2} |\log\eps|^{-1} $: Uniform Distribution of Vorticity}
\label{sec unif distr}

Before stating our results we need to introduce some notation. We define the density functional
\beq
	\label{hGPf}
	\hgpf[f] : = \int_{\ba} \diff \rv  \lf\{ \lf| \nabla f \ri|^2 - \Omega^2 r^2 f^2 + \eps^{-2} f^4 \ri\},
\eeq
for any {\it real} function $ f $. The minimization is given by
\beq
	\label{hGPe}
	\hgpe : = \inf_{f \in \hgpdom} \hgpf [f],	\hspace{1,5cm}	\hgpdom : = \lf\{ f \in H_0^1(\ba) : \: f = f^*, \lf\| f \ri\|_2 = 1 \ri\}
\eeq
and $ \hgpm $ is the associated minimizer (see Proposition \ref{hGPf minimization}). In order to give a precise meaning to the expression `bulk of the condensate', we introduce the following Thomas-Fermi functional, obtained by dropping the first term in \eqref{GPf magn} or \eqref{hGPf}: 
\beq
	\label{TFf}
	\tff[\rho] : = \frac{1}{\eps^2} \int_{\ba} \diff \rv \lf\{ \rho^2 - \eps^2 \Omega^2 r^2 \rho \ri\},
\eeq
which is expected to provide the energy associated with the non-uniform density of the condensate. We refer to the Appendix for the properties of its ground state energy $\tfe$ and associated minimizer $\tfm$. Let us define
\beq \label{TFsupport}
\tfd : = \tfsupp.
\eeq
If $\Om \leq 2(\sqrt{\pi} \ep)^{-1}$, $\tfd = \ba$, while if $\Om > 2(\sqrt{\pi} \ep) ^{-1}$, $\tfd$ is an annulus of outer radius $1$ and inner radius $\rtf$ with $1-\rtf \propto (\ep \Om ) ^{-1}$. As we shall see below, $|\gpm| ^2$ is close to $\tfm$ and thus, if $\Om \gg \ep ^{-1}$, the mass of $\gpm$ is concentrated close to the boundary of $\ba$.  
\newline
Our result about the uniform distribution of vorticity in fact holds in a slightly smaller region than $\tfd$, namely the annulus
\beq \label{bulk}
\At : = \left\{ \vec{r} \in \ba : \: \rtilde \leq r \leq \rmax \right \}
\eeq
where, for a certain quantity $ \gamma := \gamma(\eps,\Omega) >0 $ such that $ \gamma = o(1) $ as $ \eps \to 0 $ (see Section \ref{sec distribution}, Equation \eqref{rtilde} for its precise definition),
\beq
	\label{Rbulkinf}
	\rtilde := 
	\begin{cases}
		0,									&	\mbox{if} \:\: \Omega \leq \bar{\Omega} \eps^{-1}, \: \mbox{with} \:\: \bar{\Omega} < 2/\sqrt{\pi},	\\
		\rtf + \gamma \eps^{-1} \Omega^{-1},	&	\mbox{if} \:\:  2 (\sqrt{\pi} \eps)^{-1} \lesssim \Omega \ll \eps^{-2} |\log\eps|^{-1},
	\end{cases}
\eeq
and $\rmax$ is the position of the unique maximum of the density $\hgpm$ (see Proposition \ref{unique maximum}). It should be noted that $\rtilde$ is close to $\rtf$ and $\rmax$ is close to $1$ in such a way that
\bdm
	\lf| \tfd \setminus \At \ri| \ll \OO(\eps^{-1}\Omega^{-1}) = \lf| \tfd \ri|,
\edm
i.e., the domain $ \At $ tends to the support of the TF density as $ \eps \to 0 $. Also, thanks to the above estimate, we have 
\beq
\label{bulk estimate}
\int_{\At} \diff \rv \: |\gpm| ^2 = 1 - o(1),
\eeq
i.e., the mass is concentrated in $\At$. We refer to \eqref{lb max pos 1}, \eqref{lb max pos 2} and \eqref{improved rmax} below for precise estimates of $\rmax$.


\vspace{0,3cm}

We now state our result about the uniform distribution of vorticity. It is the analogue of \cite[Theorem 3.3]{CY}  but here we prove that the distribution of vorticity is uniform in the whole regime $|\log \ep| \ll \Om \ll \ep ^{-2} |\log \ep| ^{-1}$ whereas in \cite{CY} this was proved only for $\Om \lesssim \ep^{-1}$.
	
	\begin{teo}[\textbf{Uniform distribution of vorticity}]
		\label{uniform distribution}
		\mbox{}	\\
		Let $ \gpm$ be any GP minimizer and $ \eps > 0 $ sufficiently small. If $|\log \ep| \ll \Om \ll \ep ^{-2}|\log \ep| ^{-1}$, there exists a finite family of disjoint balls\footnote{Throughout the whole paper the notation $ \B(\rv,\varrho) $ stands for a ball of radius $ \varrho $ centered at $ \rv $, whereas $ \B(R) $ is a ball with radius $ R $ centered at the origin.} $ \lf\{ \ba_i \ri\} : = \{ \ba(\rvi, \varrho_i) \} \subset \At $ such that 
		\ben
			\item	$ \varrho_i \leq \OO(\Omega^{-1/2}) $, $ \sum \varrho_i \leq \OO(\Omega^{1/2}) $ and $ \sum \varrho_i^2 \ll (1 + \eps\Omega)^{-1} $, 
			\item	 $ \lf| \gpm \ri|^2 \geq \const \gamma (1 + \eps \Omega)  $ on $ \partial \ba_i $ for some $ C > 0 $.
		\een
		Moreover, denoting by $ d_{i,\eps} $ the winding number of $ |\gpm|^{-1} \gpm $ on $ \partial \ba_i $ and introducing the measure
		\beq
			\label{measure}
			\nu : = \frac{2\pi}{\Omega} \sum d_{i,\eps} \delta\lf(\rv - \rv_{i,\eps} \ri),
		\eeq
		then, for any family of sets $ \set \subset \At $ such that $ |\set| \gg \Omega^{-1} |\log(\eps^2 \Omega |\log\eps|)|^2 $ as $\ep \rightarrow 0$,
		\beq
		\label{measure convergence}
			\frac{\nu(\set)}{|\set|} \underset{\eps \to 0}{\longrightarrow} 1.
		\eeq
	\end{teo}

	\begin{rem}{\it (Distribution of vorticity)}
		\mbox{}	\\
		The result proven in the above Theorem implies that the vorticity measure converges after a suitable rescaling to the Lebesgue measure, i.e., the vorticity is uniformly distributed. However such a statement is meaningful only for angular velocities at most of order $ \eps^{-1} $, when the TF support $ \tfd $ can be bounded independently of $ \eps $. On the opposite if $ \Omega \gg \eps^{-1} $, $ \tfd $ shrinks and its Lebesgue measure converges to 0 as $ \eps^{-1} \Omega^{-1} $. To obtain an interesting statement one has therefore to allow the domain $\mathcal S$ to depend on $\eps$ with $|\mathcal S|\to 0$ as $\eps\to 0$.
	\end{rem}
	
	\begin{rem}{\it (Conditions on $ \set $)}
		\mbox{}	\\
		We remark that the lower bound on the measure of the set $ \set $, i.e., $ |\set| \gg  \Omega^{-1} |\log(\eps^2 \Omega |\log\eps|)|^2 $, is important, even though not optimal, as it will be clear in the proof: In order to localize the energy bounds to suitable lattice cells, one has to reject a certain number of \lq bad cells' where nothing can be said about the vorticity of $ \gpm $. However since the number of bad cells is much smaller than the total number of cells, this has no effect on the final statement provided the measure of $ \set $ is much larger than the area of a single cell, i.e., $ \OO(\Omega^{-1} |\log(\eps^2 \Omega |\log\eps|)|^2) $. A similar effect occurs in \cite[Theorem 3.3]{CY}, where the stronger condition $ |\set| > C $ is assumed.
	\end{rem}

	\begin{rem}{\it (Vortex balls)}
		\mbox{}	\\
		The balls contained in the family $ \{ \B_i \} $ are not necessarily vortex cores in the sense that each one might contain a large number of vortices. However the conditions stated at point 1 of the above Theorem \ref{uniform distribution} have important consequences on the properties of the family. For instance, if $ \Omega \gg \eps^{-1} $, the last one, i.e., $ \sum \varrho_i^2 \ll  \eps^{-1}\Omega^{-1} $, guarantees that the area covered by balls is smaller than the area of the annulus $ \At $ where the bulk of the condensate is contained. At the same time the other two conditions imply that the radius of any ball in the family is at most $ \OO(\Omega^{-1/2}) $ and their number can not be too large: Assuming that for each ball $ \varrho_i \sim \Omega^{-1/2} $, the second condition would yield a number of balls of order at most $ \Omega $, which is expected to be close to the total winding number of any GP minimizer.
	\end{rem}

An important difference between the flat Neumann and the flat Dirichlet problems can be seen directly from the energy asymptotics. Indeed, in the flat Neumann case (see \cite[Theorem 3.2]{CY}) the energy is composed of the contribution of the TF profile (leading order) and the contribution of a regular vortex lattice (subleading order). In the flat Dirichlet case the radial kinetic energy arising from the vanishing of the GP minimizer on $\partial\ba$ might be larger (see Remark 1.4 below) than the contribution of the vortex lattice. As a result the functional (\ref{hGPe}) that includes this radial kinetic energy plays a key role in the energy asymptotics of the problem:

	\begin{teo}[\textbf{Ground state energy asymptotics}]
		\label{gs asympt}
		\mbox{}	\\
		As $ \eps \to 0 $,
		\beq
			\label{gs en asympt 1}
			\gpe = \hgpe + \Omega |\log ( \eps^2\Omega )| (1 + o(1)),
		\eeq
		if $ |\log \ep| \ll \Omega \lesssim \eps^{-1} $, and
		\beq
			\label{gs en asympt 2}
			\gpe = \hgpe + \Omega |\log\eps| (1 + o(1)),
		\eeq
		if $ \eps^{-1} \lesssim \Omega \ll \eps^{-2} |\log\eps|^{-1} $.
	\end{teo}

	\begin{rem}{\it (Composition of the energy)}
		\mbox{}	\\
		The leading order term in the GP energy asymptotics is given by the energy $ \hgpe $ which contains the kinetic contribution of the density profile (see \eqref{hGPf}), i.e., one can decompose $ \hgpe $ as $ \tfe + \OO(\eps^{-1}) + \OO(\eps^{1/2} \Omega^{3/2}) $, where the first remainder is the most relevant in the regime $ \Omega \lesssim \eps^{-1} $ and the second becomes dominant for angular velocities much larger than $ \eps^{- 1} $. 
		\newline
		The kinetic energy of the density profile can in turn be decomposed into the energy associated with Dirichlet conditions $ \propto \eps^{-1} + \eps^{1/2}\Omega^{3/2} $ and the one due to the inhomogeneity of the profile $ \sim \sqrt{\tfm} $, which is $ \OO(1) + \OO(\eps^{2}\Omega^2|\log\eps|) $ (see Remark 2.1). The first contribution dominates for any angular velocity $ \Omega \ll \eps^{-3} |\log\eps|^{-2} $ and this is why it is the only one appearing in \eqref{gs en asympt 1} and \eqref{gs en asympt 2}. 
		\newline
		Note also that the kinetic energy due to Dirichlet boundary conditions is, in general, much larger than the vortex energy contribution, i.e., the second term in (\ref{gs en asympt 1}) and (\ref{gs en asympt 2}), except in the narrow regime
\bdm
	\eps^{-1} |\log(\eps^2\Omega)|^{-1} \ll \Omega \ll \eps^{-1} |\log\eps|,
\edm
where the latter becomes predominant.
	\end{rem}
	
An important consequence of the above energy asymptotics is that we always have (see Proposition \ref{GPmin estimates})
\beq \label{density estimate}
\left\Vert |\gpm| ^2 - \tfm \right\Vert_{L ^2(\ba)} = o(1) \ll \left\Vert \tfm \right\Vert_{L ^2 (\ba)}
\eeq	
which allows to deduce
\beq \label{bulk estimate 2}
\int_{\tfd} \diff \rv \: |\gpm| ^2 = 1 - o(1).
\eeq
This implies that if $\Om > 2 ( \sqrt{\pi} \ep) ^{-1}$, the mass of $\gpm$ is concentrated in an annulus, marking the transition to the vortex-lattice-plus-hole state. We thus have
\beq \label{result Osec}
\Osec = \frac{2}{\sqrt{\pi}\ep} (1+o(1)).
\eeq
Note that we actually prove stronger results than \eqref{bulk estimate} and \eqref{bulk estimate 2}. If $ \Omega > \Osec $, any GP minimizer is in fact exponentially small in the central hole, minus possibly a very thin layer close to $r=\rtf$ (see Proposition \ref{GPm exponential smallness}).

\subsection{The Regime $ \Omega \sim \eps^{-2} |\log\eps|^{-1} $: Emergence of the Giant Vortex}
\label{sec giant vortex}

When the angular velocity reaches the asymptotic regime $ \Omega \sim \eps^{-2} |\log\eps|^{-1} $ a transition in the GP ground state takes place above a certain threshold:  Vortices are expelled from the essential support of any GP minimizer $ \gpm $. The density is concentrated in a shrinking annulus where such a wave function is vortex free. Anticipating this transition we shall throughout this section assume that
\beq
	\label{angular velocity gv}
	\Omega = \frac{\Omega_0}{\eps^2 |\log\eps|},
\eeq
for some constant $ \Omega_0 > 0 $.
\newline
The bulk of the condensate has to be defined differently in this regime: We set
\beq
	\label{annulus}
	 \At : = \lf\{ \rv \in \ba : \: \rb \leq r \leq 1 - \eps^{3/2}|\log\eps|^2 \ri\}
\eeq
where
\beq
	\label{rd gv}
	\rb : = \rtf + \eps |\log\eps|^{-1}.
\eeq
The main result in this regime is contained in the following

	\begin{teo}[\textbf{Absence of vortices in the bulk}]
		\label{no vortices}
		\mbox{}	\\
		If the angular velocity is given by \eqref{angular velocity gv} with $ \Omega_0 > 2(3\pi)^{-1} $, then no GP minimizer has a zero inside $ \At $ if $ \eps $ is small enough.
		\newline
		More precisely, for any $ \rv \in \At $,
		\beq
			\label{pointwise est gv}
			\lf| \lf| \gpm(\rv) \ri|^2 - \tfm(r) \ri| \leq \OO(\eps^{-3/4} |\log\eps|^{2}) \ll \tfm(r).
		\eeq
	\end{teo}
	
	\begin{rem}{\it (Bulk of the condensate)}
		\mbox{}	\\
		As the notation indicates, the domain $ \At $ contains the bulk of the condensate: Using the explicit expression \eqref{TFm} of $ \tfm(r) $, one can easily verify that
		\beq
			\lf\| \tfm \ri\|_{L^2(\At)} = 1 - \OO(|\log\eps|^{-4}),
		\eeq
		which implies by \eqref{pointwise est gv} that the same estimate holds true also for $ |\gpm|^2 $.
	\end{rem}

A consequence of this result is the estimate
\beq \label{result Othird}
\Othird \leq \frac{2}{3\pi \ep ^2 |\log \ep|} (1 + o(1)).
\eeq
As already noted, we believe that this upper bound is optimal, i.e., we actually have
\[
\Othird = \frac{2}{3\pi \ep ^2 |\log \ep|}(1 + o(1)).
\]
The proof of this conjecture could use the tools of \cite{R} but we leave this aside for the present. 

\vspace{0,3cm}

The theorem above is based on a comparison of a minimizer with a giant vortex wave function of the form
	\bdm		
	f(r) \exp\lf\{i \lf( [\Omega] - \omega \ri) \vartheta\ri\},
	\edm
	where $ [ \:\:\cdot\:\: ] $ stands for the integer part and $ \omega \in \Z $ is some additional phase. Therefore we introduce a density functional
	\bml{
		\label{gvf}
		\gvf[f] : = \gpf\lf[ f(r) \exp\lf\{i \lf( [\Omega] - \omega \ri) \vartheta\ri\} \ri] = \\
		\int_{\ba} \diff \rv \lf\{ \lf| \nabla f \ri|^2 + ( [\Omega] - \omega)^2 r^{-2} f^2 - 2 ( [\Omega] - \omega) \Omega f^2 + \eps^{-2} f^4 \ri\} =	\\
		\int_{\ba} \diff \rv \lf\{ \lf| \nabla f \ri|^2 + B_{\omega}^2 f^2 - \Omega^2 r^2 f^2 + \eps^{-2} f^4 \ri\},
	}
	where $ f \in \gpdom $ is {\it real-valued} and
	\beq
		\label{rmagnp}
		\rmagnp(r) : = \lf( \Omega r - ([\Omega] - \omega) r^{-1} \ri) \vec{e}_{\vartheta}.
	\eeq
	We also set
	\beq
		\label{gve}
		\gve := \inf_{f \in \gpdom, f= f^*} \gvf[f].
	\eeq
	By simply testing the GP functional on a trial function of the form above, one immediately obtains the upper bound
	\beq
		\label{giante}
		\gpe \leq \giante : = \inf_{\omega \in \Z} \gve.
	\eeq
	In the following Theorem we prove that the r.h.s. of the expression above gives precisely the leading order term in the asymptotic expansion of $ \gpe $ as $ \eps \to 0 $ and we state an estimate of the phase optimizing $ \gve $.
	
	\begin{teo}[\textbf{Ground state energy asymptotics and optimal phase}]
		\label{giant vortex teo}
		\mbox{}	\\
		For any $ \Omega_0 > (3\pi)^{-1} $ and $ \eps $ small enough
		\beq
			\label{gv en asympt}
			\gpe = \giante + \OO((\log|\log\eps|)^{-2} |\log\eps|^2).
		\eeq
		Moreover $ \giante = \gveopt $ with $ \oopt \in \N $ satisfying
		\beq
			\label{oopt}
			\oopt := \frac{2}{3\sqrt{\pi} \eps} (1 + \OO(|\log\eps|^{-4}).
		\eeq
	\end{teo}

\begin{rem}{\it (Composition of the energy)}
		\mbox{}	\\
		We refer to \cite[Remark 1.4]{CRY} for details on the energy $\giante$ (denoted $\hgpe$ in that paper). Let us just emphasize that in this setting the Dirichlet boundary condition is responsible for a radial kinetic energy contribution that was not present in the flat Neumann case and gives the leading order correction $ \propto \eps^{-5/2}|\log\eps|^{-3/2} $ to $ \tfe $ in the asymptotic expansion of $ \giante $.
	\end{rem}

A consequence of Theorem \ref{no vortices} is that the degree of $\gpm$ is well defined on any circle $ \partial \B(r) $ of radius $ r $ centered at the origin, as long as 
\[
	\rb \leq r \leq 1- \ep ^{3/2} |\log \ep| ^2.
\] 
We are able to estimate this degree, proving that it is in agreement with that of the optimal giant vortex trial function (\ref{oopt}):

	\begin{teo}[\textbf{Degree of a GP minimizer}] 
		\label{theo degree}
		\mbox{}	\\
		If $ \Omega_0 > 2 (3\pi)^{-1} $ and $\ep$ is small enough,
			\begin{equation}\label{result degree}
				\dg \left\{ \gpm, \partial \B(r) \right\} = \Om - \frac{2}{3\sqrt{\pi} \eps} (1 + \OO(|\log\eps|^{-4}),
			\end{equation} 
		for any $ \rb \leq r \leq 1- \ep ^{3/2} |\log \ep| ^2$.
	\end{teo}

We note that because of the Dirichlet condition there is a small region close to $\dd \ba$ where the density goes to zero. We have basically no information on the GP minimizer in this layer that could a priori contain vortices. The existence of this layer is the main difference between the flat Dirichlet case and the flat Neumann case considered in \cite{CRY}. In particular the lack of a priori estimates on the phase circulation of $\gpm$ on $\dd \ba$ requires new ideas in the proof. 

\subsection{Rotational Symmetry Breaking}
\label{sec symm break}

As anticipated above, a very natural question arising from the results in Section 1.2 is that of the repartition of vortices in the central hole of low matter density. In particular, does one have 
\[
\gpm  = \gvmopt (r) e^{i\left( \Om- \oopt \right)\vartheta},
\] 
modulo a constant phase factor, which would imply that all the vorticity is contained in a central multiply quantized vortex?\\
We show below that this can not be the case: the GP functional is rotationally symmetric but if the angular velocity exceeds a certain threshold this symmetry is broken at the level of the ground state. No minimizer of the GP energy functional is an eigenfunction of the angular momentum, i.e. a function of the form $f(r) e^{in\vartheta}$ with $f$ real and $n$ an integer. A straightforward consequence is that there is not a unique minimizer but for any given minimizing function one can obtain infinitely many others by simply rotating the function by an arbitrary angle. In other words as soon as the rotational symmetry is broken, the ground state is degenerate and its degeneracy is infinite.

\vspace{0,3cm}

In \cite[Proposition 2.2]{CDY1} we have proven that the symmetry breaking phenomenon occurs in the case of a bounded trap $ \ba $ with Neumann boundary conditions when $ c |\log\eps| \leq \Omega \lesssim \eps^{-1} $, for some given constant $ c $. We are now going to show that such a result admits an extension to angular velocities much larger than $ \eps^{-1} $, i.e., the rotational symmetry is still broken even for very large angular velocities. Such an extension is far from obvious in view of the main result about the emergence of a giant vortex state discussed above: Since vortices are expelled from the essential support of the GP minimizer, there might a priori be a restoration of the rotational symmetry but the behavior of any GP minimizer inside the hole $ \ba(\rtf) $ remains unknown. 

	\begin{teo}[\textbf{Rotational symmetry breaking}]
		\label{symmetry breaking}
		\mbox{}	\\
		If $ \eps$ is small enough and $\ep \Om$ large enough, no minimizer of the GP energy functional \eqref{GPf} is an eigenfunction of the angular momentum.
	\end{teo}


We note that it is proved in \cite{AJR} for a related model that the ground state \emph{is} rotationally symmetric if $\Om < \Ofirst$ and $\ep$ is small enough. Theorem \ref{symmetry breaking} shows that the symmetry, broken due to the nucleation of vortices, \emph{never} reappears, even when $\Om > \Othird$.

\subsection{Organization of the Paper}

The paper is organized as follows. Section \ref{sec preliminary} is devoted to general estimates that will be used throughout the paper. We then prove our results about the regime $|\log \ep| \ll \Om \ll \ep ^{-2} |\log \ep| ^{-1}$ in Section \ref{sec unif distr proof}. The analysis of the energy functional \eqref{hGPf} is the main new ingredient with respect to the method of \cite{CY}. We adapt the techniques developed in that paper for the evaluation of the energy of a trial function containing a regular lattice of vortices. The corresponding lower bound is proved via a localization method allowing to appeal to results from GL theory \cite{SS1,SS2}. The inhomogeneity of the density profile is dealt with using a Riemann sum approximation.\\
Section \ref{sec gvortex} is devoted to the giant vortex regime. Our main tools are the techniques of vortex ball construction and jacobian estimates, originating in the papers \cite{Sa,J,JS} (see also \cite{SS2}). We implement this approach using a cell decomposition as in \cite{CRY}. New ideas are necessary to control the behavior of GP minimizers on $\dd \ba$.\\ 
The symmetry breaking result is proved in Section \ref{sec symm break proof}. Following \cite{Seir}, given a candidate rotationally symmetric minimizer, we explicitly construct a wave function giving a lower energy. Finally the Appendix gathers important but technical results about the TF functional and the third critical speed.

\section{Preliminary Estimates: The Density Profile with Dirichlet Boundary Conditions}
\label{sec preliminary}

This section is devoted to the proof of estimates which will prove to be very useful in the rest of the paper but are independent of the main results. We mainly investigate the properties of the density profile which captures the main traits of the modulus of the GP minimizer $ |\gpm|$: More precisely we study in details the minimization of the density functional $ \hgpf $ \eqref{hGPf} and prove bounds on its ground state energy $ \hgpe $ \eqref{hGPe} and associated minimizers $ \hgpm $. 

The leading order term in the ground state energy $ \hgpe $ is given by the infimum of the TF functional \eqref{TFf}, i.e., 
\beq
	\label{TFe}
	\tfe : = \inf_{\rho \in \tfdom} \tff[\rho],	\hspace{1,5cm}	\tfdom : = \lf\{ \rho \in L^1(\ba) : \: \rho > 0, \lf\| \rho \ri\|_1 = 1 \ri\}.
\eeq
We postpone the discussion of the properties of $ \tfe $ as well as the corresponding minimizer $ \tfm $ to the Appendix.

\begin{pro}[\textbf{Minimization of $ \hgpf $}]
	\label{hGPf minimization}
	\mbox{}	\\
	If $ \Omega \ll \eps^{-3}|\log\eps|^{-2} $ as $ \eps \to 0 $,
	\beq
		\label{hGPe asympt}
		\tfe \leq \hgpe \leq \tfe + \OO(\eps^{-1}) + \OO\lf(\eps^{1/2} \Omega^{3/2} \ri).
	\eeq
	Moreover there exists a minimizer $g$ that is unique up to a sign, radial and can be chosen to be positive away from the boundary $\partial\ba$. It solves inside $ \ba $ the variational equation
	\beq
		\label{hGPm variational eq}
		- \Delta \hgpm - \Omega^2 r^2 \hgpm + 2 \eps^{-2} \hgpm^3 = \hchem \hgpm,
	\eeq
	with boundary condition $ \hgpm(1) = 0 $ and $ \hchem = \hgpe + \eps^{-2} \lf\| \hgpm \ri\|_4^4 $.
\end{pro}

\begin{rem}{\it (Composition of the energy $ \hgpe $)}
	\label{rem energy comp}
	\mbox{}	\\
	The remainders appearing on the r.h.s. of \eqref{hGPe asympt} can be interpreted as the kinetic energy due to Dirichlet boundary conditions: The bending of the TF density close to $ r = 1 $ in order to fulfill the boundary condition produces some kinetic energy which is not negligible and can be estimated by means of the trial function used in the proof of the above proposition, i.e., $ \OO(\eps^{-1}) $ as long as $ \Omega\lesssim \eps^{-1} $, and $ \OO(\eps^{1/2}\Omega^{3/2}) $ for larger angular velocities. Note indeed that the second correction becomes relevant only if $ \Omega \gtrsim \eps^{-1} $.
\newline
The orders of those corrections can be explained as follows: If $ \Omega \lesssim \eps^{-1} $ the TF density goes from its maximum of order 1 to 0 in a layer of thickness $ \sim \eps $ (because of the nonlinear term), yielding a gradient $ \sim \eps^{-1} $ and thus a kinetic energy of order $ \eps^{-1} $. If $ \Omega \gg \eps^{-1} $ the thickness of the annulus where $ \hgpm $ varies from $ \sqrt{\eps\Omega} $ to 0 becomes of order $ \eps^{1/2} \Omega^{-1/2} $ and the associated kinetic energy is $ \OO(\eps^{1/2}\Omega^{3/2}) $.
\newline
Note that in both cases the kinetic energy associated with the boundary conditions is much larger than the radial kinetic energy of the profile $ \sqrt{\tfm} $ which is $ \OO(1) $ in the first case and $ \OO(\eps^2\Omega^2|\log\eps|) $ in the second one (see \cite[Section 4]{CY}): The condition $ \Omega \ll \eps^{-3} |\log\eps|^{-2} $ is precisely due to the comparison of such energies for large angular velocities.
\newline
Finally we point out that, if $ \Omega \ll \eps^{-1} $, the correction of order $ \eps^{-1} $ due to Dirichlet boundary conditions can become much larger than two terms of order $ \Omega^2 $ and $ \eps^2 \Omega^4 $  contained inside $ \tfe $ (see the explicit expression \eqref{TFe explicit} in the Appendix), so that the upper bound could be stated in that case $ \hgpe \leq \pi^{-1} \eps^{-2} + \OO(\eps^{-1}) $.
\end{rem}

\begin{proof}[Proof of Proposition \ref{hGPf minimization}]
	\mbox{}	\\
	The lower bound is trivial since it is sufficient to neglect the positive kinetic energy to get $ \hgpe \geq \tfe $.
	\newline
	The upper bound is obtained by evaluating $ \hgpf $ on a trial function of the form
	\begin{equation}
		\label{trial profile}
	 	\ftrial(r) = c \sqrt {\rho (r)} \xi_{D}(r)
	\end{equation}
	where $c$ is the normalization constant and $ 0 \leq \xi_{D}(r) \leq 1 $ a cut-off function equal to 1 everywhere except in the radial layer $ [1-\delta,1] $, $ \delta \ll (1 + \eps\Omega)^{-1} $, where it goes smoothly to 0, so that $ f $ satisfies Dirichlet boundary conditions. The density $\rho (r)$ coincides with $ \tfm(r) $ if $ \Omega $ is below the threshold $ 2 (\sqrt{\pi}\eps)^{-1} $ and is given by a regularization of $ \tfm $ above it, i.e., if $ \eps \Omega > 2/\sqrt{\pi} $, we set as in \cite[Eq. (4.9)]{CY}
	\begin{equation}
		\label{trial density}
		\rho (r) : = 
			\begin{cases} 
					0,									& \text{if } \: 0 \leq r \leq \rtf,	\\
					\Omega^2 \tfm(\rtf + \Omega^{-1}) (r - \rtf)^2,	&	\text{if } \: \rtf \leq r \leq \rtf + \Omega^{-1},		\\
					\tfm(r),	 							&	\text{otherwise}.
		\end{cases}
	\end{equation}
	Notice that $ \rho $ differs from $ \tfm $ only inside the interval $ [\rtf, \rtf + \Omega^{-1}] $ and
	\beq
		\label{trial point est}
		 \rho(r) = \tfm(r) + \OO(\eps^2 \Omega).
	\eeq
	In order to estimate the normalization constant we use the bound $ \tfm \leq C(1 + \eps\Omega)$, which implies
	\beq
		\label{norm const}
		c^{-2} = \int_{\ba} \diff \rv \: \rho \xi_{D}^2 \geq \int_{\ba}\diff \rv \: \tfm - 2\pi \int_{1-\delta}^{1} \diff r \: r  (1 - \xi_{D}^2) \rho - C \eps^2 \geq 1 - C[(\eps\Omega + 1)\delta + \eps^2].
	\eeq
	The kinetic energy of $ \ftrial $ is bounded as follows:
	\begin{equation}
		\label{kinetic trial}
		\int_{\ba} \diff \rv \: \lf| \nabla \ftrial \ri|^2 \leq 2 c^2 \int_{\ba} \diff \rv \: \lf\{ \lf| \nabla \sqrt{\rho} \ri|^2 + \rho \lf| \nabla \xi_{D} \ri|^2 \ri\} \leq C[ \varepsilon^2 \Omega^2 | \log \varepsilon| + (\eps\Omega + 1) \delta^{-1}],
	\end{equation}
	where we refer to \cite[Eqs. (4.14) and (4.15)]{CY} for the estimate of the kinetic energy of $ \rho $.
	\newline
	The interaction term can be easily estimated as
	\beq
		\label{interaction trial}
		\frac{1}{\eps^2} \int_{\ba} \diff \rv \: \ftrial^4 \leq \frac{1 + C(\eps\Omega + 1) \delta}{\varepsilon^2}  \lf\| \tfm \ri\|_2^2  \leq \eps^{-2}  \lf\| \tfm \ri\|_2^2 + C (\Omega + \eps^{-1})^2 \delta.
	\eeq
	To evaluate the centrifugal term we act as in \cite[Eqs. (4.44) -- (4.46)]{CY}: With analogous notation
	\bml{
		\label{centrifugal trial}
		- \Omega^2 \int_{\ba} \diff \rv \: r^2 \ftrial^2 = - 2 \pi \Omega^2 + 4 \pi \Omega^2 \int_0^1 \diff r \: r \int_0^r \diff r^{\prime} \: r^{\prime} \ftrial^2(r^{\prime}) \leq 	\\
		- 2 \pi \Omega^2 + 4 \pi \Omega^2 \int_0^1 \diff r \: r \int_0^r \diff r^{\prime} \: r^{\prime} \tfm(r^{\prime}) + C [\Omega + (\eps^2 + \Omega^{-2})^{-1} \delta] =	\\
		- \Omega^2 \int_{\ba} \diff \rv \: r^2 \tfm + C (\Omega + \eps^{-2} \delta + \Omega^2 \delta),
	}
	where we have integrated by parts twice and used \eqref{trial point est}, \eqref{norm const} and the normalization of $ \ftrial $. Hence one finally obtains
	\beq
		\hgpf\lf[\ftrial\ri] \leq \tfe + C[ \varepsilon^2 \Omega^2 | \log \varepsilon| + \Omega + (\eps\Omega + 1) \delta^{-1}  + (\Omega + \eps^{-1})^2 \delta].
	\eeq
	It only remains to optimize w.r.t. $ \delta $, which yields $ \delta = \eps $, if $ \Omega \lesssim \eps^{-1} $, and $ \delta = \eps^{1/2} \Omega^{-1/2} $ otherwise, and thus the result.
\end{proof}

A crucial property of the density $ \hgpm $ is stated in the following 

\begin{pro}[\textbf{Behavior of $ \hgpm $}]
	\label{unique maximum}
	\mbox{}	\\
	The density $ \hgpm $ admits a unique maximum at some point $ 0 < \rmax < 1 $.
\end{pro}

\begin{proof} 
	The method is very similar to what is used in \cite[Lemma 2.1]{CRY}, although in that case one considers the Neumann problem. After a variable transformation $r^2 \rightarrow s$ the functional $ \hgpf $ becomes
	\beq
		\pi \int^1_0 \diff s \left \{ s| \nabla f|^2 - \Omega^2 s f^2 + \eps^{-2} f^4 \right \},
\end{equation}
and the normalization condition
	\beq
 		\int^1_0 \diff s \: g^2 = \pi^{-1}.
	\eeq
	We first observe that the Dirichlet boundary condition implies that $g$ cannot be constant, otherwise we would have $ \hgpm = 0 $ everywhere, contradicting the mass constraint.
	\newline
	Suppose now that  $ \hgpm $ has more than one local maximum. Then it has a local minimum at some point $ s = s_2 $ with $ 0 < s_2 < 1 $, on the right side of a local maximum at the position $ s = s_1 $, i.e., $ s_1 < s_2 $. For $ 0 < \epsilon < \hgpm^2 (s_1) - \hgpm^2 (s_2)$, we consider the set $\mathcal I_\epsilon = \{ 0 \leq s < s_2: \: \hgpm^2(s_1) - \epsilon \leq \hgpm^2(s) \leq \hgpm^2 (s_1) \}$: Since $ \hgpm $ is continuous, the function 
	\beq
		\Phi(\epsilon) := \int_{\mathcal I_\epsilon} \diff s \: \hgpm^2
	\eeq
	is strictly positive and $ \Phi(\epsilon) \rightarrow 0 $ as $ \epsilon \rightarrow 0$. Likewise, for any $ \kappa > 0 $, we set $ \mathcal J_\kappa = \{ s_1 < s \leq 1: \: \hgpm^2(s_2) \leq \hgpm^2(s) \leq \hgpm^2 (s_2) + \kappa \}$, so that 
 	\beq
 		\Gamma(\kappa) : = \int_{\mathcal J_\kappa} \diff s \: \hgpm^2
 	\eeq
  	has the same properties as $ \Phi$. 
	\newline
	Hence, by the continuity of $\hgpm $, there always exist $\overline \epsilon, \overline \kappa > 0$, such that $\hgpm^2(s_2) + \overline \kappa < \hgpm^2(s_1) - \overline \epsilon$ and  $ \Phi(\overline \epsilon) = \Gamma(\overline \kappa)$. Note that this implies that $\mathcal I_{\overline \epsilon}$ and $\mathcal J_{\overline \kappa}$ are disjoint. 
	\newline
	We now define a new normalized function $ \tilde{g} $ by
	\beq
		\tilde{g}^2 (s) : = 
		\begin{cases} 
				\hgpm^2 (s_1)^2 - \overline \epsilon,		&	\text{if } s \in \mathcal I_{\overline \epsilon},	\\ 
				\hgpm^2(s_2)+ \overline \kappa, 			&	\text{if }s \in \mathcal J_{\overline \kappa}, \\
				\hgpm^2(s),						&	\text{otherwise}.
		\end{cases}
	\eeq
	The gradient of $ \tilde{g} $ vanishes in the intervals $\mathcal I_{\overline \epsilon}$ and $ \mathcal J_{\overline \kappa}$ and equals the gradient of $ \hgpm $ everywhere else, so that the kinetic energy of $ \tilde{g} $ is smaller or equal to the one of $ \hgpm $. The centrifugal term is lowered by $ \tilde{g} $, because $ -s $ is strictly decreasing and the value of $ \tilde{g}^2 $ on $\mathcal I_{\overline \epsilon}$ is larger than on $ \mathcal J_{\overline \kappa}$. Finally since mass is rearranged from $\mathcal I_{\overline \epsilon}$ to $J_{\overline \kappa}$, where the density is lower, $ \lf\| \tilde{g} \ri\|_4^4 < \lf\| \hgpm \ri\|_4^4 $.
	\newline
	Therefore the functional evaluated on $ \tilde{g} $ is strictly smaller than $ \hgpe $, which contradicts the assumption that $ \hgpm $ is a minimizer. Hence $ \hgpm $ has only one maximum.
\end{proof}

The energy asymptotics \eqref{hGPe asympt} implies that the density $ \hgpm^2 $ is close to the TF minimizer $ \tfm $:

	\begin{pro}[\textbf{Preliminary estimates of $ \hgpm $}]
		\label{preliminary est hgpm}
		\mbox{}	\\
		If $ \Omega \ll \eps^{-3} |\log\eps|^{-2} $ as $ \eps \to 0 $,
		\beq
			\label{hGPm l2 est}
			\lf\| \hgpm^2 - \tfm \ri\|_{L^2(\ba)} \leq \OO(\eps^{1/2} + \eps^{5/4} \Omega^{3/4}),	
		\eeq
		\beq
			\label{hGPm linfty est}
			\hgpm^2(\rmax) = \lf\| \hgpm \ri\|^2_{L^{\infty}(\ba)} \leq \lf\| \tfm \ri\|_{L^{\infty}(\ba)}  \lf(1 +\OO(\sqrt{\eps}) + \OO(\eps^{3/4} \Omega^{1/4}) \ri).
		\eeq
	\end{pro}
	
	\begin{proof}
		See, e.g., \cite[Proposition 2.1]{CRY}. Note that \eqref{hGPm l2 est} implies the bound
		\beq
			\label{chemical diff}
			\lf| \hchem - \tfchem \ri| \leq C \lf( \eps\Omega + 1 \ri)^{1/2} \lf( \eps^{-3/2} + \eps^{-3/4} \Omega^{3/4} \ri),
		\eeq
		which yields \eqref{hGPm linfty est}.
	\end{proof}

Next proposition is going to be crucial in the proof of the main results since it allows to replace the density $ \hgpm^2 $ with the TF density $ \tfm $: On the one hand, using the fact that the latter is explicit, this result will be used to obtain a suitable lower bound on $ \hgpm^2 $ in some region far from the boundary and, on the other hand, it implies that the boundary layer where $ \hgpm $ goes to 0 is very small.

	\begin{pro}[\textbf{Pointwise estimate of $ \hgpm $}]
		\label{pointwise bound dens}
		\mbox{}	\\		
		If $ \Omega \lesssim \bar{\Omega} \eps^{-1} $ with $ \bar{\Omega} < 2/\sqrt{\pi} $ as $ \eps \to 0 $,
		\beq
			\label{pointwise bound 1}
			\lf| \hgpm^2(r) - \tfm(r) \ri| \leq \OO(\sqrt{\eps}),
		\eeq
		for any $ 0 \leq r \leq 1 - \OO(\eps|\log\eps|) $. 
		\newline
		On the other hand\footnote{This second estimate applies also if $ \Omega = 2(\sqrt{\pi}\eps)^{-1} (1 - o(1)) $, in which case $ \rtf $ has to be set equal to 0.} if $ 2(\sqrt{\pi}\eps)^{-1} \lesssim \Omega \lesssim \eps^{-2} $,
		\beq
			\label{pointwise bound 2}
			\lf| \hgpm^2(r) - \tfm(r) \ri| \leq \OO(\eps^{7/4} \Omega^{5/4}),
		\eeq
		for any $ \rtf + \OO(\eps^{-1} \Omega^{-1} |\log\eps|^{-2}) \leq r \leq 1 - \OO(\eps^{1/2}\Omega^{-1/2} |\log\eps|^{3/2}) $.
	\end{pro}

	\begin{rem}{\it (Position of the maximum of $ \hgpm $)}
		\label{maximum pos rem}
		\mbox{}	\\
		The pointwise estimates \eqref{pointwise bound 1} and \eqref{pointwise bound 2} give some information about the position of the maximum point of $ \hgpm $. Assuming that $ \Omega \lesssim \eps^{-1} $, one has the lower bound
		\bdm
			\hgpm^2(\rmax) \geq \tfm(1 - \eps |\log\eps|) - C \sqrt{\eps},
		\edm
		since \eqref{pointwise bound 1} holds true up to a distance $ \eps|\log\eps| $ from the boundary. Hence one immediately obtains 
		\beq
			\label{lb max pos 1}
			\rmax \geq 1 - \OO(\eps^{-3/2}\Omega^{-2}),
		\eeq
		which for $ \eps^{-3/4} \ll \Omega \lesssim \eps^{-1} $ implies that $ \rmax = 1 - o(1) $. For smaller angular velocities the above inequality becomes useless: Since $ \tfm $ is approximately constant in those regimes, i.e., $ \tfm(r) = \pi^{-1} + \OO(\eps^2\Omega^2) $, the pointwise estimate \eqref{pointwise bound 1} is too rough to extract information about  the maximum of $ \hgpm $.
		\newline
		On the opposite if  $ 2(\sqrt{\pi}\eps)^{-1} \lesssim \Omega \lesssim \eps^{-2} $, we get from \eqref{pointwise bound 2} the following: Either $ \rmax \geq 1 -  \eps^{1/2}\Omega^{-1/2} |\log\eps|^{3/2} $ or the pointwise estimate \eqref{pointwise bound 2} applies at $ \rmax $ yielding $ \hgpm^2(\rmax) \leq \tfm(\rmax) + C \eps^{7/4} \Omega^{5/4}  $ and
		\bdm
			\hgpm^2(\rmax) \geq \tfm(1 - \eps^{1/2} \Omega^{-1/2} |\log\eps|^{3/2}) - C \eps^{7/4} \Omega^{5/4},
		\edm
		so that, in any case,
		\beq
			\label{lb max pos 2}
			\rmax \geq 1 - \OO(\eps^{-1/4} \Omega^{-3/4}),
		\eeq
		since $ \eps^{1/2}\Omega^{-1/2} |\log\eps|^{3/2} \ll \eps^{-1/4} \Omega^{-3/4} $.
	\end{rem}

	\begin{rem}{\it (Improved pointwise estimates of $ \hgpm $)}
		\mbox{}	\\
		Thanks to the remark above, it is possible to refine the estimates \eqref{pointwise bound 1} and \eqref{pointwise bound 2} and extend them up to the maximum point of $ \hgpm $. More precisely one has the following: If $ \Omega \lesssim \bar{\Omega} \eps^{-1} $ with $ \bar{\Omega} < 2/\sqrt{\pi} $,
		\beq
			\label{ref pointwise bound 1}
			\lf| \hgpm(r) - \tfm(r) \ri| \leq \OO(\sqrt{\eps}),
		\eeq
		for any $ 0 \leq r \leq \max[\rmax, \: 1 - \eps|\log\eps|] $. If $ 2(\sqrt{\pi}\eps)^{-1} \lesssim \Omega \lesssim \eps^{-2} $,
		\beq
			\label{ref pointwise bound 2}
			\lf| \hgpm(r) - \tfm(r) \ri| \leq \OO(\eps^{7/4} \Omega^{5/4}),
		\eeq
		for any $ \rtf + \eps^{-1} \Omega^{-1} |\log\eps|^{-2} \leq r \leq \max[\rmax, \: 1 - \eps^{1/2} \Omega^{-1/2} |\log\eps|^{3/2}] $.
		\newline
		The extension can be easily done in the first case ($ \Omega \lesssim \bar{\Omega} \eps^{-1} $) by noticing that one can suppose $ \rmax \geq 1 - \eps|\log\eps| $ (otherwise the bound is given by the original result), so that \eqref{ref pointwise bound 1} follows from \eqref{hGPm linfty est} together with
		\bdm
			\lf\| \tfm \ri\|_{\infty} - \tfm(1 - \eps|\log\eps|) = \tfm(1) - \tfm(1 - \eps|\log\eps|) \leq C \eps^{3} \Omega^2 |\log\eps| \ll \OO(\sqrt{\eps}),
		\edm	
		and the fact that $ \hgpm $ is increasing in $ \ba(\rmax) $. 
		\newline
		In the other regime the key point is the estimate \eqref{lb max pos 2}, which implies 
		\bdm
			\max_{1 - \eps^{-1/4} \Omega^{-3/4} \leq r \leq 1} \lf| \tfm(1) - \tfm(r) \ri| \leq C \eps^{7/4} \Omega^{5/4}.
		\edm
	\end{rem}


\begin{proof}[Proof of Proposition \ref{preliminary est hgpm}]
	\mbox{}	\\
	The proof is done exactly as in \cite[Proposition 2.6]{CRY}, so we highlight only the main differences.
	\newline
	The result is obtained by exhibiting suitable local sub- and super-solutions to the variational equation
	\beq
		-\Delta \hgpm = 2 \eps^{-2} \lf( \hgpd - \hgpm^2 \ri) \hgpm,
	\eeq
	where the function $ \hgpd $ is given by
	\beq
		 \label{hGPdd}
		\hgpd(r) : = \half \lf( \eps^2 \hchem + \eps^2 \Omega^2 r^2 \ri).
	\eeq
	\newline
	By \eqref{chemical diff}, if $ \Omega \leq 2(\sqrt{\pi}\eps)^{-1}$,
	\beq
		\label{diff densities 1}
		\lf\| \hgpd- \tfm \ri\|_{L^{\infty}(\ba)} \leq C \sqrt{\eps},
	\eeq
	whereas if $ \Omega > 2(\sqrt{\pi}\eps)^{-1} $,
	\beq
		\label{diff densities 2}
		\lf| \hgpd(r) - \tfm(r) \ri| \leq C \eps^{7/4} \Omega^{5/4},
	\eeq
	for any $ r \geq \rtf $.
	\newline
	In order to apply the maximum principle one needs a lower bound on the function $ \hgpd $ in the domain under consideration and it is provided by the above estimates: In the fist case, i.e., if $ \Omega \leq \bar{\Omega} \eps^{-1} $, $ \tfm(r) \geq C(\bar{\Omega}) > 0 $ and the pointwise estimate \eqref{diff densities 1} guarantees the positivity of $ \hgpd $ everywhere; otherwise, if $  2(\sqrt{\pi}\eps)^{-1} \lesssim \Omega  \lesssim \eps^{-2} $, the condition $ r \geq \rtf + \OO(\eps^{-1} \Omega^{-1} |\log\eps|^{-2}) $ yields
	\bdm
		\tfm(r) \geq C \eps \Omega |\log\eps|^{-2} \gg \OO( \eps^{7/4} \Omega^{5/4}) \geq \lf| \hgpd(r) - \tfm(r) \ri|,
	\edm
	so that $ \hgpd(r) > C \eps\Omega |\log\eps|^{-2} > 0 $ in the region considered.
	\newline
	The rest of the proof is done as in \cite[Proposition 2.6]{CRY} in a local annulus $ [r_0 - \delta, r_0 + \delta] $ with $ \delta  = \eps |\log\eps| $, if  $ \Omega \lesssim \bar{\Omega} \eps^{-1} $, and $ \delta = \eps^{1/2} \Omega^{-1/2} |\log\eps|^{3/2} $ otherwise. Note that the lack of monotonicity of the density profile $ \hgpm $ prevents a straightforward extension of the estimate to the whole support of $ \tfm $.
\end{proof}

For angular velocities larger than the threshold $ 2(\sqrt{\pi}\eps)^{-1} $ the TF density develops a hole centered at the origin of radius $ \rtf $ (see the Appendix) and in this case one can show that the density $ \hgpm $ is exponentially small there:

\begin{pro}[\textbf{Exponential smallness of $ \hgpm $ inside the hole}]
		\label{exponential smallness}
		\mbox{}	\\
		If $ \Omega \geq 2(\sqrt{\pi}\eps)^{-1} $ as $ \eps \to 0 $,
		\beq
			\label{exp small}
			\hgpm^2(r) \leq C \eps \Omega \: \exp \lf\{ - \frac{1-r^2}{1 - \rtf^2} \ri\}
		\eeq 
		 for any $ \rv \in \B $. Moreover, if $ \Omega \geq 2(\sqrt{\pi}\eps)^{-1} + \OO(1) $,  there exist a strictly positive constant $ c $ such that for any $ \rv $ such that $ r \leq \rtf - \OO(\eps^{7/6}) $,
		\beq
			\label{improved exp small}
			\hgpm^2(r)  \leq C  \eps \Omega  \: \exp \lf\{ - \frac{c}{\eps^{1/6}} \ri\}.
		\eeq
	\end{pro}

	\begin{proof}
		See \cite[Proposition 2.2]{CRY}. Note that in the proof of the second estimate, the condition $  \Omega \geq 2(\sqrt{\pi}\eps)^{-1} + \OO(1) $ is needed in order to guarantee that $ \rtf \gg \OO(\eps^{7/6}) $.
	\end{proof}

	The pointwise estimates \eqref{ref pointwise bound 1} and \eqref{ref pointwise bound 2} and the exponential smallness stated in the proposition above have some important consequences as, e.g., an improved $L^2$ estimate on the density $ \hgpm $ close to the boundary of the trap:
	
	\begin{pro}[\textbf{Estimate of $ \rmax $ and $ L^2 $ estimate of $ \hgpm $}]
		\label{improved rmax est}
		\mbox{}	\\
		If $ \eps^{-1} \lesssim \Omega \ll \eps^{-2} $,
		\beq
			\label{improved rmax}
			\rmax \geq 1 - \OO(\eps^{-5/8}\Omega^{-7/8}),		\hspace{1,5cm}	\lf\| \hgpm \ri\|^2 _{L^2(\ba\setminus\ba(\rmax))} \leq \OO(\eps^{3/8}\Omega^{1/8}).
		\eeq
	\end{pro}

	\begin{proof}
		Without loss of generality we can assume $ \Omega > 2(\sqrt{\pi}\eps)^{-1} $, since the proof in the other case, i.e., without the hole, is even simpler. Because of the normalization of both $ \tfm $ and $ \hgpm $, we have
		\bdm
			\int_{\ba\setminus\ba(\rmax)} \diff \rv \lf( \tfm - \hgpm^2 \ri)= \int_{\ba(\rmax)} \diff \rv  \lf( \hgpm^2 - \tfm \ri).
		\edm
		The monotonicity of $ g $ in $ \ba(\rmax) $ and the bound \eqref{ref pointwise bound 2} yield $ \hgpm^2(\rtf) \leq  \hgpm^2(\rtf + \eps^{-1} \Omega^{-1} |\log\eps|^{-1}) \leq \eps\Omega |\log\eps|^{-1} $, so that, setting for convenience $ R_0 : = \rtf + \eps^{-1} \Omega^{-1} |\log\eps|^{-1} $ and using the exponential smallness \eqref{improved exp small}, one obtains
		\bmln{
			\int_{\ba(R_0)} \diff \rv  \lf| \tfm - \hgpm^2 \ri| \leq \int_{ \ba(\rtf - \eps^{7/6})} \diff \rv \: \hgpm^2 + C \eps^{1/6} \Omega^{-1} |\log\eps|^{-1} + \int_{\ba(R_0) \setminus \ba(\rtf) } \diff \rv  \lf| \tfm - \hgpm^2 \ri| \leq	\\
			C \lf\| \tfm - \hgpm^2 \ri\|_{L^2(\ba)} \eps^{-1/2} \Omega^{-1/2} |\log\eps|^{-1/2} +  C \eps^{1/6} \Omega^{-1} |\log\eps|^{-1} \leq C \eps^{3/4} \Omega^{1/4} |\log\eps|^{-1/2}.
		}
		For $ r \geq R_0 $ one can apply the pointwise estimate \eqref{ref pointwise bound 2}, which yields
		\bdm
			\int_{\ba(\rmax) \setminus \ba(R_0)} \diff \rv \:  \lf| \tfm - \hgpm^2 \ri| \leq C \eps^{3/4} \Omega^{1/4}.
		\edm
		Collecting the above estimates one therefore has
		\beq
			\label{l1 est boundary}
			\int_{\ba\setminus\ba(\rmax)} \diff \rv  \lf( \tfm - \hgpm^2 \ri) \leq C \eps^{3/4} \Omega^{1/4}.
		\eeq
		On the other hand by \eqref{ref pointwise bound 2}, $ \hgpm^2(\rmax) \leq \tfm(\rmax) + C \eps^{7/4} \Omega^{5/4}) $, so that
		\bmln{
 			\int_{\ba\setminus\ba(\rmax)} \diff \rv  \lf( \tfm - \hgpm^2 \ri) \geq \frac{\eps^2\Omega^2}{2} \int_{\ba\setminus\ba(\rmax)} \diff \rv \lf( r^2 - \rmax^2 \ri) - C \eps^{7/4} \Omega^{5/4} (1 - \rmax^2) \geq 	\\
			\mbox{$\frac{1}{4}$} \pi \eps^2\Omega^2 (1 - \rmax^2)^2 - C \eps^{7/4} \Omega^{5/4} (1 - \rmax^2),
		}
		which gives the estimate of $ \rmax $.
		\newline
		Since the argument leading to \eqref{l1 est boundary} is symmetric in $ \hgpm^2 $ and $ \tfm $, it is also true that
		\beq
		\label{l2 norm outside}
			\int_{\ba\setminus\ba(\rmax)} \diff \rv \: \hgpm^2 \leq \int_{\ba\setminus\ba(\rmax)} \diff \rv \: \tfm + C \eps^{3/4} \Omega^{1/4} \leq C \eps^{3/8} \Omega^{1/8},
		\eeq
		due to the lower bound on $ \rmax $ \eqref{improved rmax}.
	\end{proof}

\section{The Regime $ |\log\eps| \ll \Omega \ll \eps^{-2}|\log\eps|^{-1} $}
\label{sec unif distr proof}

This section contains the proof of the main results stated in the Introduction for the regime $ |\log\eps| \ll \Omega \ll \eps^{-2}|\log\eps|^{-1} $. We also prove some additional estimates, which are basically corollaries of the main results and will be used also in the analysis of the giant vortex regime.

\subsection{GP Energy Asymptotics}

The most important result proven in this section is the GP ground state energy asymptotics:

\begin{proof}[Proof of Theorem \ref{gs asympt}]
	\mbox{}	\\
	The result is proven by exhibiting upper and lower bounds for the GP ground state energy.

\indent \emph{Step 1.} For the upper bound we evaluate the GP functional on the trial function
	\beq
		\label{trial  function}
		\trial(\rv) : = c \: \hgpm(r) \xi(\rv) \Phi(\rv),
	\eeq
	where $ c $ is a normalization constant, $ \Phi(\rv) $ the phase factor introduced in \cite[Eq. (4.6)]{CY} and $ \xi $ a cut-off function: More precisely, using the complex notation $ \zeta = x + iy \in \mathbb{C} $ for points $ \rv = (x,y) \in \RR $, we can express $ \Phi $ as 
	\beq
		\Phi(\rv) : = \prod_{\zeta_i \in \latt} \frac{\zeta - \zeta_i}{|\zeta - \zeta_i|},
	\eeq
	where we denote by $ \latt $ a finite, regular lattice (triangular, rectangular or hexagonal) of points $ \rv_i \in \ba $ such that the corresponding cell $ \celli $ is contained in $ \ba $: Each lattice point  $\rv_i$ lies at  the center of a lattice cell $ \celli $ and the lattice constant $\ell$ is chosen so that the area of  the fundamental cell $ \cell $ is 
	\beq
		|\cell|= \pi \Omega^{-1}.
	\eeq
	Thus $ \ell = C \Omega^{-1/2} $ and the total number of lattice points in the unit disc is $  \Omega (1 - \OO(\Omega^{-1/2})) $. In order to get rid of the singularities of the phase factor $ \Phi $ at lattice points, we define the function
	\begin{equation}
		\xi(\rv) : = 
			\begin{cases} 
				1,				&	\text{if} \:\: |\zeta - \zeta_i| > t, \: \forall \zeta_i \in \latt,		\\
 				t^{-1} |\zeta - \zeta_i|, 	&	\text{if} \:\: |\zeta - \zeta_i| \leq t, \mbox{ for some } \zeta_i \in \latt
			\end{cases}
	\end{equation}
	where $t$ is a variational parameter satisfying the conditions $ \min [\eps, \eps^{1/2} \Omega^{-1/2} ] \leq t \ll \Omega^{-1/2}$.
	\newline
	The normalization constant takes into account the effect of the cut-off function $ \xi $ and it is not difficult to see that $ 1 \leq c^2 \leq 1 + C \Omega t^2 $.
	\newline
	We start by computing the kinetic energy of $ \trial $:
	\beq
		\label{kinetic energy trial}
 		\int_{\ba} \diff \rv \lf| \lf(\nabla - i  \magnp\ri) \trial \ri|^2 = c^2 \int_{\ba} \diff \rv \: \lf| \nabla \lf(\hgpm \xi\ri)\ri|^2 + c^2 \int_{\ba} \diff \rv \: \xi^2 \hgpm^2 \lf| \nabla \Phi - \magnp \ri|^2.
	\eeq
	The first term in the expression above can be estimated as follows:
	\bml{
 		\label{radial kinetic}
		c^2 \int_{\ba} \diff \rv \: \lf| \nabla \lf(\hgpm \xi\ri)\ri|^2 - \int_{\ba} \diff \rv \: \lf| \nabla \hgpm \ri|^2 \leq \half c^2 \int_{\ba} \diff \rv \: \nabla \hgpm^2 \cdot \nabla \xi^2 + C \lf( \Omega + \Omega \eps^{-1} t^2 + \eps^{1/2} \Omega^{5/2} t^2 \ri) \leq	\\
		\half c^2 \sum_{\rvi \in \latt} \int_{\partial \ba(\rv_i,t)} \diff \sigma \: g^2 \partial_n \xi^2 - \half c^2 \int_{\ba} \diff \rv \: g^2 \Delta \xi^2 + C \lf( \Omega + \Omega \eps^{-1} t^2 + \eps^{1/2} \Omega^{5/2} t^2 \ri) \leq	\\
		C \lf( \Omega + \Omega \eps^{-1} t^2 + \eps^{1/2} \Omega^{5/2} t^2 \ri)
	}
where we have used the bounds $ | \nabla \xi |\leq t^{-1} $, $ |\Delta \xi^2| \leq C t^{-2} $ and $ \lf\| \nabla \hgpm \ri\|^2_2 \leq C (\eps^{-1} + \eps^{1/2} \Omega^{3/2}) $ (see \eqref{hGPe asympt}). We have also used the fact that
	\bdm
		\sum_{\rvi \in  \latt} \sup_{\rv \in \ba(\rvi,t)} g^2(r) \leq C \Omega,
	\edm
	which can be seen as a consequence of the upper bound $ g^2 \leq C(\eps\Omega + 1) $ in addition to the exponential smallness \eqref{improved exp small}, which allows to estimate the above quantity as the number of cells contained in $ \tfsupp $ times $ \eps\Omega + 1 $, i.e., $ \OO(\Omega) $.
	\newline
	In order to estimate the last term in \eqref{kinetic energy trial}, we act exactly as in \cite[Proposition 4.1]{CY}. The estimate (4.37) in \cite{CY}, that is obtained by making use of an analogy with an electrostatic problem, reads in our case
	\beq
		\int_{\ba} \diff \rv \: \xi^2 \hgpm^2 \lf| \nabla \Phi - \magnp \ri|^2 \leq (1 + C t \Omega^{1/2}) \sum_{\rvi \in \latt} \sup_{\rv \in \celli} \hgpm^2(r) \lf( \pi |\log(t^2\Omega)| + \OO(1) \ri).
	\eeq
	It remains then to use the Riemann sum approximation and the normalization of $ \hgpm^2 $ to estimate the sum in the above expression: If $ \Omega \leq \bar{\Omega} \eps^{-1} $ for some $ \bar{\Omega} < 2/\sqrt{\pi} $, we can simply use \eqref{ref pointwise bound 1} to replace $ \hgpm^2 $ with $ \tfm $ and proceed as in the proof of Proposition 4.1 in \cite{CY}, obtaining
	\beq
		\label{ub kinetic energy 1} 
		\int_{\ba} \diff \rv \: \xi^2 \hgpm^2 \lf| \nabla \Phi - \magnp \ri|^2 \leq \lf[ 1 + \OO( t \Omega^{1/2}) +  \OO(\sqrt{\eps}) \ri] |\cell|^{-1} \lf( \pi |\log(t^2\Omega)| + \OO(1) \ri).
	\eeq
	Note that inside each cell $ \sup \tfm - \inf \tfm \leq C \eps^2 \Omega^{3/2} \ll \sqrt{\eps} $, so this error term can be absorbed in the $\OO (\sqrt{\eps})$ in the equation above.
	\newline
	In the opposite case, if $ \Omega \geq 2(\sqrt{\pi}\eps)^{-1} $, we set
	\beq
	\label{domain}
		\D : = \lf\{ \rv \in \ba : \: r \geq \rd \ri\},
	\eeq
	with
	\beq
		\label{rd}
		\rd := \rtf + \eps^{-1} \Omega^{-1} |\log(\eps^2 \Omega |\log\eps|)|^{-1},
	\eeq
	so that $ \rd - \rtf \ll \eps^{-1} \Omega^{-1} $ and 
	\beq
		\label{TFd lower bound}
		\tfm(r) \geq \half \eps\Omega |\log\eps|^{-1} (1 - o(1)), \qquad \forall \rv \in \D,
	\eeq
since $ |\log(\eps^2 \Omega |\log\eps|)| \leq |\log\eps|(1 + o(1)) $. 
	\newline
	Now we can replace $ \hgpm^2 $ with $ \tfm $ inside $ \D $ by means of \eqref{ref pointwise bound 2}. Moreover in the region $ r \leq \rd $ we can use the exponential smallness \eqref{improved exp small}, if $ r \leq \rtf - \eps^{7/6} $, and the pointwise bound $ g^2(r) \leq g^2 (\rd) \leq C \eps \Omega |\log\eps|^{-1} $, if $ \rtf - \eps^{7/6} \leq r \leq \rd $, which follows from \eqref{ref pointwise bound 2} and the monotonicity of $ g^2(r) $ in $ \ba(\rmax) $. The result is the upper estimate
	\bml{
 		\label{ub kinetic energy 2} 
 		\int_{\ba} \diff \rv \: \xi^2 \hgpm^2 \lf| \nabla \Phi - \magnp \ri|^2 \leq \lf[1 + C( t \Omega^{1/2} + \eps^{3/4} \Omega^{1/4} |\log\eps|) \ri] \sum_{\rvi \in \latt \cap \D} \sup_{\rv \in \celli} \tfm(r) \lf( \pi |\log(t^2\Omega)| + \OO(1) \ri) +	\\
		C \eps \Omega |\log\eps|^{-1}  \lf| \{ \rv : \: \rtf - \eps^{7/6} \leq r \leq \rd \} \ri| \lf| \cell \ri|^{-1} |\log\eps| + C \eps \Omega |\log\eps| \exp\lf\{ - c \eps^{-1/6} \ri\} \leq	\\
		 \lf[1 + C( t \Omega^{1/2} + \eps^{3/4} \Omega^{1/4} |\log\eps|) \ri] \sum_{\rvi \in \latt} \sup_{\rv \in \celli} \tfm(r) \lf( \pi |\log(t^2\Omega)| + \OO(1) \ri) + C \Omega |\log\eps|^{-1}	\leq	\\
		 \lf[1 + C( t \Omega^{1/2} + \eps^{3/4} \Omega^{1/4} |\log\eps| + \eps \Omega^{1/2}) \ri] \lf| \cell\ri|^{-1} \lf( \pi |\log(t^2\Omega)| + \OO(1) \ri) + C \Omega |\log\eps|^{-1},
	}
	where we have used the estimate $ \sup \tfm - \inf \tfm \leq C \eps^2 \Omega^{3/2} $ inside any cell $ \celli \subset \D $.
	\newline
	The estimate of $ \tff[\left| \trial \right| ^2] $ can be obtained as in \cite[Eqs. (4.42) and (4.48)]{CY} (see also \eqref{centrifugal trial}):
	\bml{
 		\label{diff TF energies}
 		\tff \lf[ \lf| \trial \ri|^2 \ri] \leq (1 + C \Omega t^2) \eps^{-2} \int_{\ba} \diff \rv \: \hgpm^4 - \Omega^2 \int_{\ba} \diff \rv \: r^2 \hgpm^2 + C \eps^{-1} \Omega^{2} t^2 \leq	\\
		\tff\lf[ \hgpm^2 \ri] + C \lf[ \eps^{-2} \Omega t^2 +  \eps^{-1} \Omega^2 t^2 \ri].
	}
	To conclude the proof of the upper bound it only remains to choose the variational parameter $ t $: In the regime $ \Omega \leq \bar{\Omega} \eps^{-1} $, $ \bar\Omega < 2/\sqrt{\pi} $, we take $ t = \eps $ so that the remainder occurring in the above estimate becomes $ \OO(\Omega) $ as in \eqref{radial kinetic} and \eqref{ub kinetic energy 1}, whereas, if $ \Omega \geq 2 (\sqrt{\pi}\eps)^{-1} $, the remainder in \eqref{diff TF energies} leads to $ t^2  = \eps \Omega^{-1} $ in order to recover the same error term $ \OO(\Omega) $ as in \eqref{radial kinetic}. In \eqref{ub kinetic energy 1} there is an additional remainder of order $ \OO(\eps\Omega^{3/2} |\log\eps|) $ which might become larger than $ \Omega $ for very large angular velocities and is due to the Riemann sum approximation.
	\newline
	The final result is therefore
	\beq
		\label{GP energy ub 1}
		\gpf\lf[\trial\ri] \leq \hgpe + \Omega |\log(\eps^2\Omega)| + \OO(\Omega),
	\eeq
	if $ 1 \ll \Omega \lesssim \eps^{-1} $, and
	\beq
		\label{GP energy ub 2}
		\gpf\lf[\trial\ri] \leq \hgpe + \Omega |\log\eps| + \OO(\Omega) + \OO(\eps\Omega^{3/2} |\log\eps|),
	\eeq
	if $ \Omega \ll \eps^{-2} $.

\indent \emph{Step 2.} The starting point of the lower bound proof is a decoupling of the energy which can be obtained by defining a function $ u(\rv)  $ as
	\beq
		\label{function u}
		u(\rv) : = \hgpm^{-1}(r) \gpm(\rv).
	\eeq
	Note that, thanks to the positivity of $ g $, the function $ u $ is well defined in the open ball $ \{ \rv : \: r < 1 \} $.
	\newline\
	By means of this definition and the variational equation \eqref{hGPm variational eq}, one can use a method originating in \cite{LM} to decouple the energy (see, e.g., \cite[Proposition 3.1]{CRY} or \cite[Lemma 2.2]{Se} where a Dirichlet boundary condition also appears) to obtain, using the $L^2$ normalization of both $ \gpm $ and $ \hgpm $, 
	\[
	\gpf[\gpm] = \hgpe + \int_{\ba} \diff \rv \: \hgpm^2 \lf\{ \lf| (\nabla - i \magnp) u \ri|^2 + \eps^{-2} \hgpm^2 \lf( 1 - |u|^2 \ri)^2 \ri\}.
	\]
	We deduce the lower bound
	\beq
		\label{energy decoupling}
		\gpe = \gpf[\gpm] \geq \hgpe + \int_{\Dt} \diff \rv \: \hgpm^2 \lf\{ \lf| (\nabla - i \magnp) u \ri|^2 + \eps^{-2} \hgpm^2 \lf( 1 - |u|^2 \ri)^2 \ri\}
	\eeq
	 by restricting the last integral to $ \Dt $, with 
	\beq
		\Dt : =
		\begin{cases}
		\lf\{ \rv \in \ba : \: r \leq 1 - \eps|\log\eps| \ri\},	&	\mbox{if} \:\: \Omega \leq \bar{\Omega} \eps^{-1}, \: \mbox{with} \:\: \bar{\Omega} < 2/\sqrt{\pi},	\\
		\ba(\rmax) \cap \D					&	\mbox{if} \:\: \Omega \gtrsim 2(\sqrt{\pi}\eps)^{-1}.
		\end{cases}
	\eeq
	The pointwise estimates \eqref{ref pointwise bound 1} and \eqref{ref pointwise bound 2} allow the replacement of $ \hgpm^2 $ with $ \tfm $:
	\beq
		\label{energy decoupling 2}
		\gpe \geq \hgpe + \lf[ 1 - C \lf( \sqrt{\eps} + \eps^{3/4} \Omega^{1/4} \ri) \ri] \int_{\Dt} \diff \rv \: \tfm \lf\{ \lf| (\nabla - i \magnp) u \ri|^2 + \eps^{-2} \tfm \lf( 1 - |u|^2 \ri)^2 \ri\}.
	\eeq
 Moreover as in \cite[Section 5]{CY} we define another regular (square) lattice
	\beq
		\hlatt : = \lf\{ \rv_i  = (m \hspac, n \hspac), m,n \in \Z : \: \hcelli \subset \Dt \ri\},
	\eeq
	where $ \hcelli $ is the cell centered at $ \rv_i $ and the lattice spacing satisfies the same conditions as in \cite[Eq. (5.16)]{CY}, i.e.,
	\beq
		\label{lattice spacing cond}
		|\log\eps|^{1/2} \Omega^{-1/2} \ll \hspac \ll \min \lf[1, (\eps\Omega)^{-1} |\log(\eps^2\Omega |\log\eps|)|^{-1} \ri],
	\eeq
	so that
	\bdm
		\sup_{\rv \in \celli} \lf| \tfm(r) - \tfm(r_i) \ri| \leq C \eps \Omega \hspac |\log\eps| \tfm(r_i).
	\edm
	Hence \eqref{energy decoupling 2} yields the lower bound
	\bml{
 		\label{GPe lower bound}
		\gpe - \hgpe \geq \lf[ 1 - \OO(\sqrt{\eps}) - \OO(\eps^{3/4} \Omega^{1/4} |\log\eps|) - \OO(\eps \Omega \hspac |\log\eps|) \ri] \sum_{\rvi \in \hlatt} \tfm(r_i) \E^{(i)}[u] \geq	\\
		(1 - o(1)) \sum_{\rvi \in \hlatt} \tfm(r_i) \E^{(i)}[u],
	}
	where $ \E^{(i)} $ is defined as in \cite[Eq. (5.18)]{CY}, i.e.,
	\beq
		\E^{(i)}[u]:= \int_{\hcelli} \diff \rv \lf\{ \lf| (\nabla - i\magnp) u \ri|^2 + \eps^{-2} \tfm(r_i) \lf(1 - |u|^2\ri)^2 \ri\}.
	\eeq
After a suitable scaling the energy above can be seen as a Ginzburg-Landau energy with a fixed external field $h_{\rm ex}$ in the
range $|\log \epsilon| \ll h_{\rm ex} \ll \epsilon^{-2} $ where $\epsilon$ is a new small parameter. We can thus use the lower bound for the Ginzburg-Landau energy (see \cite{SS1,SS2}) as in \cite[Proposition 5.1]{CY}. The result is
	\beq
		\label{Ei lower bound}
		\E^{(i)}[u] \geq \Omega \hspac^2 |\log(\min[\eps,\eps^2\Omega])| (1 - o(1)),
	\eeq
	for any $ |\log\eps| \ll \Omega \ll \eps^{-2} |\log\eps|^{-1} $.
	\newline
	To complete the proof if suffices then to use, for any $ \Omega \ll \eps^{-1} $,  the estimate $ \tfm(r) \geq \pi^{-1} (1 - o(1)) $,  which yields
	\beq
		\label{Riemann sum 1}
		\sum_{\rvi \in \hlatt} \tfm(r_i) \geq (1 - o(1)) \pi^{-1} (1 - \OO(\eps|\log\eps|)) |\ba| |\hcell|^{-1} \geq (1 - o(1)) \hspac^{-2},
	\eeq 
	and thus the result. On the other hand, if $ \Omega \gtrsim \eps^{-1} $, a simple computation (see, e.g., \eqref{l2 norm outside}) using the estimates \eqref{lb max pos 1} and \eqref{improved rmax} gives
	\beq
		\lf\| \tfm \ri\|_{L^1(\ba\setminus\ba(\rmax))} \leq o(1),
	\eeq
	which implies
	\beq
		\label{Riemann sum 2}
		\sum_{\rvi \in \hlatt} \tfm(r_i) \geq (1 - o(1)) \hspac^{-2} \int_{\ba(\rmax)\cap \D} \diff \rv \: \tfm(r) \geq (1 - o(1))  \hspac^{-2},
	\eeq
	thanks to the normalization of $ \tfm $.
	\newline
	By putting together \eqref{GPe lower bound}, \eqref{Ei lower bound}, \eqref{Riemann sum 1} and \eqref{Riemann sum 2}, one obtains the lower bound matching \eqref{GP energy ub 1} and \eqref{GP energy ub 2}.
\end{proof}

\subsection{Estimates for GP Minimizers}

The GP energy asymptotics has many important consequences on the asymptotic behavior of GP minimizers: For instance the upper bounds \eqref{GP energy ub 1} and \eqref{GP energy ub 2} immediately imply the $L^2$ convergence of any minimizing density $ |\gpm|^2 $ to the TF density $ \tfm $:

	\begin{pro}[\textbf{$L^2$ convegence of $ |\gpm|^2 $}]
		\label{GPmin estimates}
		\mbox{}	\\
		As $ \eps \to 0 $, if $ |\log\eps| \ll \Omega \lesssim \eps^{-1} $,
		\beq
			\label{GPm l2 estimate 1}
			\lf\| |\gpm|^2 - \tfm \ri\|_{L^2(\ba)} \leq \OO(\eps^{1/2}) + \OO(\eps \Omega^{1/2} |\log(\eps^2\Omega)|^{1/2}),
		\eeq
		whereas, if $ \eps^{-1} \ll \Omega \ll \eps^{-2} $,
		\beq
			\label{GPm l2 estimate 2}
			\lf\| |\gpm|^2 - \tfm \ri\|_{L^2(\ba)} \leq \OO(\eps \Omega^{1/2} |\log\eps|^{1/2}) + \OO(\eps^{5/2} \Omega^{3/2}).
		\eeq
	\end{pro}

	\begin{proof}
		See \cite[Proposition 2.1]{CRY}.
	\end{proof}

Acting as in the derivation of \cite[Eq. (2.8)]{CRY}, one can show that the above $L^2$ estimates imply a bound on the chemical potential $ \chem $ occurring in the variational equation \eqref{GP variational} solved by $ \gpm $:
\beq
	\label{chemical difference 1}
	\lf| \chem - \tfchem \ri| \leq \OO(\eps^{-3/2}) + \OO(\eps^{-1/2} \Omega^{1/2} |\log(\eps^2\Omega)|^{1/2}),
\eeq
if $ |\log\eps| \ll \Omega \lesssim \eps^{-1} $, while, for $ \eps^{-1} \ll \Omega \ll \eps^{-2} $,
\beq
	\label{chemical difference 2}
	\lf| \chem - \tfchem \ri| \leq \OO(\eps \Omega^2) + \OO(\eps^{-1/2} \Omega |\log\eps|^{1/2}).
\eeq

Such estimates can in turn be used to prove a pointwise upper bound for $ |\gpm|^2 $ (see \cite[Proposition 2.1]{CRY}), i.e.,
\beq
	\label{GPm linfty estimate}
	\lf\| \gpm \ri\|_{L^{\infty}(\ba)}^2 \leq \tfm(1) \cdot
	\begin{cases}
		1 + \OO(\eps^{1/2}) + \OO(\eps^{3/2} \Omega^{1/2} |\log(\eps^2\Omega)|^{1/2}),	&	\mbox{if} \:\: |\log\eps| \ll \Omega \lesssim \eps^{-1},	\\
		1 + \OO(\eps^2 \Omega) + \OO(\eps^{1/2} |\log\eps|^{1/2}),					&	\mbox{if} \:\: \eps^{-1} \ll \Omega \ll \eps^{-2}.
	\end{cases}
\eeq

We finally state another very useful pointwise estimate of $ \gpm $ analogous to \cite[Proposition 2.2]{CRY} and Proposition \ref{exponential smallness}. As is the case for the density profile $ \hgpm $, if the angular velocity is above the threshold $ 2 (\sqrt{\pi} \eps)^{-1} $, any GP minimizer is exponentially small inside the hole $ \ba(\rtf) $.
	
	\begin{pro}[\textbf{Exponential smallness of $ \gpm $ inside the hole}]
		\label{GPm exponential smallness}
		\mbox{}	\\
		If $ \Omega \geq (2/\sqrt{\pi}) \eps^{-1} + \OO(1) $, as $ \eps \to 0 $, there exists a strictly positive constant $ c $ such that for any $ \rv $ such that $ r \leq \rtf - \OO(\eps^{7/6}) $,
		\beq
			\label{GPm exp small}
			\lf| \gpm(\rv) \ri|^2  \leq C \eps \Omega  \: \exp \lf\{ - \frac{c}{\eps^{1/6}} \ri\}.
		\eeq
	\end{pro}

	\begin{proof}
		See \cite[Proposition 2.2]{CRY}.
	\end{proof}

\subsection{Distribution of Vorticity}
\label{sec distribution}

We are now able to prove the uniform distribution of vorticity:

\begin{proof}[Proof of Theorem \ref{uniform distribution}]
	\mbox{}	\\
	The proof follows very closely the proof of \cite[Theorem 3.3]{CY} and relies essentially on \cite[Proposition 5.1]{SS1}. 
	
	The argument has to be slightly adapted depending on the value of the angular velocity: For any $ \Omega \leq \bar{\Omega} \eps^{-1} $, with $ \bar{\Omega} < 2/\sqrt{\pi} $, the proof of \cite[Theorem 3.3]{CY} applies with only one minor modification, since the cells in the lattice $ \hlatt $ occurring in the lower bound proof do not cover the whole of $ \ba $. However, since the region covered by cells tends to $ \tfd $ as $ \eps \to 0 $ and the area of the excluded set close to the boundary is of order $ \OO(\eps|\log\eps|) $, i.e., much smaller than the cell area, such a difference in the lattice choice has no consequences for the final result. 

 	We now discuss the modifications in the regime $  \eps^{-1} \ll \Omega \ll \eps^{-2} |\log\eps|^{-1} $ which was not taken into account in \cite[Theorem 3.3]{CY}. The starting point is the localization of the energy bounds \eqref{GP energy ub 2}, \eqref{GPe lower bound} and \eqref{Ei lower bound}, which can be rewritten as
	\beq
		\label{bound localization}
		\sum_{\rv_i \in \hat{\latt}} \tfm(r_i) \lf| \gpfi[u_{}] - \Omega \hat{\spac}^2 |\log\eps| \ri| \leq \eta \: \Omega \hat{\spac}^2 |\log\eps|  \sum_{\rv_i \in \hat{\latt}} \tfm(r_i),
	\eeq
	for some 
	\[
	 \eta = \eta(\eps,\Omega) \ll 1 \mbox{ as } \eps \to 0 .
	 \] 
	In order to obtain a similar estimate inside one lattice cell, one first needs a suitable lower bound on the density $ \tfm $ and this can be obtained by restricting the analysis to the bulk of the condensate, i.e.,
	\bdm 
		\At : = \left\{ \vec{r} \in \ba : \: \rtilde \leq r \leq \rmax \right \}
	\edm
	where, if $ \Omega \gg \eps^{-1} $, $ \rtilde $ is given by
	\beq
		\label{rtilde}
		\rtilde : = \rtf + \gamma \eps^{-1} \Omega^{-1},	\hspace{1,5cm} \gamma : = |\log\eta|^{-1}.
	\eeq
	We then have, for some $ C > 0 $,
	\beq
		\label{density lb inside bulk}
		\tfm(r) \geq C \eps \Omega |\log\eta|^{-1} \mbox{  on  } \At.
	\eeq
	Moreover, the localization of the energy estimate requires that a certain number of bad cells be rejected: As in \cite[Theorem 3.3]{CY} we first introduce a new small parameter
	\beq
		\epsilon : = \sqrt{\frac{2\eps |\log\eps|}{\Omega}} \ll 1,
	\eeq
	so that $ |\log(\epsilon^2\Omega)| = |\log\eps|(1 + o(1)) $ and \eqref{bound localization} yields
	\beq
		\label{GPfi ub}
		\sum_{\rv_i \in \hat{\latt}} \tfm(r_i) \gpfi_{\epsilon}[u_{}] \leq (1 + \eta) \sum_{\rv_i \in \hat{\latt}} \tfm(r_i)  \Omega \hat{\spac}^2 |\log(\epsilon^2\Omega)|,
	\eeq 
	where
	\beq
		\gpfi_{\epsilon}[u] := \int_{\hcelli} \diff \rv \lf\{ \lf| (\nabla - i\magnp) u \ri|^2 + \epsilon^{-2} \lf(1 - |u|^2\ri)^2 \ri\},
	\eeq
	with $ \eta(\epsilon,\Omega) \to 0 $ as $ \epsilon \to 0 $. 
	\newline
	We then say that a cell $ \hcelli \subset \At $ is a good cell if
	\beq
		\label{good cells unif vort}
		\gpfi_{\epsilon}[u] \leq (1 + \sqrt{\eta} ) \Omega \hat{\spac}^2 |\log(\epsilon^2\Omega)|,
	\eeq
	whereas the cell is bad if the inequality is reversed. 
	\newline
	Now given any set $ \set \subset \At $ such that $ |\set| \gg |\hat\cell| $, the upper bound \eqref{GPfi ub}, the definition of bad cells, \eqref{density lb inside bulk} and the upper bound $ \tfm \leq \OO(\eps \Omega) $ imply that 
	\beq
		N_{B} \leq \sqrt{\eta} \: \gamma^{-1} N = \sqrt{\eta} |\log\eta| N \ll N,
	\eeq 
	where $ N_{B} $ and $ N $ stand for the number of bad cells and the total number of cells contained inside $ \set $ respectively.
	\newline
	On the other hand by the definition \eqref{good cells unif vort} good cells satisfy the assumptions of \cite[Proposition 5.1]{SS1} and therefore one can construct a finite collection of disjoint balls $ \{ \ba_i \} : = 
\{ \ba(\rvi,\varrho_i) \} $ such that $ |u| > 1/2 $ on the boundary of each ball and $ \varrho_i \leq \OO(\Omega^{-1/2}) $. Hence one can define the winding number $ d_{i,\eps} $ of $ u $ on $ \partial \ba_i $, which coincides with the winding number of $ \gpm $ and, using \cite[Proposition 5.1]{SS1}
	\beq
		2 \pi \sum d_{i,\eps} = \Omega \hspac^2 (1 + o(1)),	\hspace{1,5cm}	2 \pi \sum \lf|d_{i,\eps}\ri| = \Omega \hspac^2 (1 + o(1)).
	\eeq
	The rest of the statement of Theorem \ref{uniform distribution} easily follows by noticing that one can always take $ \hspac = \Omega^{-1/2} |\log(\eps^2 \Omega |\log\eps|)| $, which satisfies \eqref{lattice spacing cond}, obtaining the lower condition on the area of the set $ \set $.
\end{proof}

\section{The Giant Vortex Regime $ \Omega \sim \eps^{-2}|\log\eps|^{-1} $}
\label{sec gvortex}

As a preparation for the proof of the main results contained in Theorems \ref{no vortices} and \ref{giant vortex teo} we formulate and prove in Section \ref{giant vortex densities sec} some important propositions about the properties of the giant vortex density profiles.
\newline
The proof of the absence of vortices in the bulk will follow the analysis of the ground state energy asymptotics, which is achieved in several steps. The main ingredients are the energy decoupling (Section \ref{energy decoupling sec}), the vortex ball construction and the jacobian estimate (Section \ref{vortex ball construction sec}). Each individual step is analogous to the corresponding one contained in \cite{CRY} and we will often omit some details, only stressing the major differences with the analysis of \cite{CRY} and referring to that paper for further details.

\subsection{Giant Vortex Density Profiles}
\label{giant vortex densities sec}

In this section we investigate the properties of the giant vortex profiles and the associated energy functional defined in \eqref{gvf}. Actually for technical reasons which will be clearer later we consider a functional identical to \eqref{gvf} but on a different integration domain, i.e.,
\beq
	\label{annulus at}
	\ann : = \lf\{ \rv \in \ba : \: r \geq \rt \ri\},
\eeq
where $ \rt < \rtf $ is suitably chosen in order to apply some estimates: All the conditions on $ \rt $ occurring in the subsequent proofs are satisfied if we take
\beq
	\label{rt}
	\rt : = \rtf - \eps^{8/7}.
\eeq
More precisely we define
\begin{equation}\label{domain tilde}
\gpdomt := \left\{  f \in H^1 (\ann) : \: f= f ^*, \: \lf\| f \ri\|_{L^2 (\ann)} = 1, \: f= 0 \mbox{ on } \dd \ba \right\}
\end{equation}
and set, for any $ f \in \gpdomt $,
\bml{
	\label{hgvf}
	\hgvf[f] : = \int_{\ann} \diff \rv \lf\{ \lf| \nabla f \ri|^2 + ( [\Omega] - \omega)^2 r^{-2} f^2 - 2 ( [\Omega] - \omega) \Omega f^2 + \eps^{-2} f^4 \ri\} = 	\\
	\int_{\ann} \diff \rv \lf\{ \lf| \nabla f \ri|^2 + B_{\omega}^2 f^2 - \Omega^2 r^2 f^2 + \eps^{-2} f^4 \ri\}.
}
We recall that 
\[
B_{\omega} =\Omega r - \left( [ \Omega] -\omega \right) r ^{-1}.
\]
The associated ground state energy is
\beq
	\label{hgve}
	\hgve := \inf_{f \in \gpdomt} \hgvf[f]
\eeq
and we denote $ \hgvm $ any associated minimizer.
\newline
The TF-like functional obtained from \eqref{hgvf} by dropping the kinetic term is denoted by $ \htff_{\omega} $ (see \eqref{hTFf}) and its minimization discussed in the Appendix.

	\begin{pro}[\textbf{Minimization of $ \hgvf $}]
		\label{hgvf minimization}
		\mbox{}	\\
		If $ \Omega \propto \eps^{-2} |\log\eps|^{-1} $ and $ |\omega| \leq \OO(\eps^{-5/4}|\log\eps|^{-3/4}) $ as $ \eps \to 0 $, then
		\beq
			\label{hgved asympt}
			\tfe \leq \htfe_{\omega} \leq \hgve \leq \htfe_{\omega} + \OO(\eps^{-5/2} |\log\eps|^{-3/2}) \leq \tfe + \OO(\eps^{-5/2} |\log\eps|^{-3/2}).
		\eeq	
		There exists a minimizer $ \hgvm $ that is unique up to a sign, radial and can be chosen to be positive away from the boundary $\partial\ba$. It solves inside $ \ann $ the variational equation
		\beq
			\label{hgvm variational eq}
			- \Delta \hgvm + B_{\omega}^2 \hgvm - \Omega^2 r^2 \hgvm + 2 \eps^{-2} \hgvm^3 = \hgvchem \hgvm,
		\eeq
		with boundary conditions $ \hgvm(1) = 0 $ and $ \hgvm^{\prime}(\rt) = 0 $ and $ \hgvchem = \hgve + \eps^{-2} \lf\| \hgvm \ri\|_4^4 $. 
		\newline
		Moreover $ \hgvm $ has a unique global maximum at $  \rmaxgv $ with $ \rt < \rmaxgv < 1 $.
	\end{pro}

	\begin{rem}\textit{(Composition of the energy)}
		\mbox{}	\\
		Unlike the flat Neumann case, the remainder in the r.h.s. of \eqref{hgved asympt} is of the same order even if the refined TF energy $ \htfe_{\omega} $ is extracted. The reason is that such a remainder is actually due to the radial kinetic energy of the giant vortex density profile and in particular to Dirichlet boundary conditions.
		\newline
		In order to give some heuristics to explain such a difference with the flat Neumann case, it is indeed sufficient to note that, by the pointwise estimate \eqref{pointwise bound 2}, the density $ \hgvm $ goes from its maximum value $ \sim \eps^{1/2} \Omega^{1/2} \sim \eps^{-1/2} |\log\eps|^{-1/2} $ to $ 0 $ in a region of width at most $ \OO( \eps^{1/2} \Omega^{-1/2} |\log\eps|^{3/2} ) = \OO(\eps^{3/2} |\log\eps|^2) $. This yields an estimate for the kinetic energy of $ \hgvm $ in that region as $ \OO(\eps^{-5/2} |\log\eps|^{-3}) $, i.e., approximately the same remainder as in \eqref{hgved asympt}, which is in any case much larger than the difference between the TF energies $ \htfe_{\omega} - \tfe $ (see \eqref{hTFe}).
	\end{rem}

	\begin{proof}[Proof of Proposition \ref{hgvf minimization}]
		\mbox{}	\\
		The proof of Proposition \ref{hGPf minimization} applies to the functional $ \hgvf $ as well by noticing that
		\beq
			\lf| B_{\omega}(r) \ri| \leq \Omega \sup_{\rv \in \ann} \lf( r^{-1} - r \ri) + C \lf| \omega \ri| \leq C \lf( \eps^{-1} + \lf| \omega \ri| \ri),
		\eeq
		which implies that the $ B_{\omega}^2$ term in the functional (see the second expression in \eqref{hgvf}) is always smaller than the remainder in \eqref{hGPe asympt}, provided $ |\omega| \leq \OO(\eps^{-5/4}|\log\eps|^{-3/4}) $. The Neumann condition at the inner boundary of $ \ann $ is a direct consequence of the assumption $ f \in H^1(\A) $.
	\end{proof}

Since the asymptotic behavior of the energy $ \hgve $ is the same as that of $ \hgpe $ (see \eqref{hGPe asympt}) for any $ |\omega| \leq \OO(\eps^{-5/4}|\log\eps|^{-3/4}) $, most of the estimates proven for the density profile $ \hgpm $ hold true for $ \hgvm $ as well, provided the phase $ \omega $ satisfies the estimate required in Proposition \ref{hgvf minimization}. We sum up such estimates in the following

	\begin{pro}[\textbf{Estimates for $ \hgvm $}]
		\label{estimates gv dens}
		\mbox{}	\\	
		If $ \Omega \sim \eps^{-2} |\log\eps|^{-1} $ and $ |\omega| \leq \OO(\eps^{-5/4}|\log\eps|^{-3/4}) $ as $ \eps \to 0 $,
		\beq
			\label{preliminary est hgvm 1}
			\lf\| \hgvm^2 - \tfm \ri\|_{L^2(\ann)} \leq \OO( \eps^{-1/4}|\log\eps|^{-3/4}),					\eeq
		\beq
			\label{preliminary est hgvm 2}
			\hgvm^2(\rmaxgv) = \lf\| \hgvm \ri\|^2_{L^{\infty}(\ann)} \leq \lf\| \tfm \ri\|_{L^{\infty}(\ba)}  \lf(1 + \OO(\eps^{1/4} |\log\eps|^{-1/4}) \ri).
		\eeq	
		Moreover for any $ \rv \in \ann $ such that $ \rtf + \OO(\eps|\log\eps|^{-1}) \leq r \leq \max[\rmaxgv, 1 - \eps^{3/2} |\log\eps|^2] $
		\beq
			\label{pointwise bound hgvm}
			\lf| \hgvm^2(r) - \tfm(r) \ri| \leq \OO(\eps^{-3/4}|\log\eps|^{-5/4}) \leq \OO(\eps^{1/4} |\log\eps|^{7/4}) \tfm(r) \ll \tfm(r),
		\eeq
		and the maximum position $ \rmaxgv(\omega) $ of $ \hgvm $ satisfies the bounds
		\beq
			\label{improved rmaxgv}
			\rmaxgv(\omega) \geq 1 - \OO(\eps^{9/8}|\log\eps|^{7/8}),		\hspace{1,5cm}	\lf\| \hgvm \ri\|^2 _{L^2(\ba\setminus\ba(\rmaxgv))} \leq \OO(\eps^{1/8}|\log\eps|^{1/8}).
		\eeq
		Finally for any $ \rv $ such that $ r \leq \rtf - \OO(\eps^{7/6}) $,
		\beq
			\label{exp small hgvm}
			\hgvm^2(r)  \leq C  \eps^{-1} |\log\eps|^{-1} \: \exp \lf\{ - \frac{c}{\eps^{1/6}} \ri\}.
		\eeq
	\end{pro}

	\begin{proof}
		The results are proven exactly as the analogous statements contained in Propositions \ref{preliminary est hgpm}, \ref{pointwise bound dens}, \ref{exponential smallness} and \ref{improved rmax est}.
	\end{proof}

\subsection{Energy Decoupling and Optimal Phases}
\label{energy decoupling sec}

The first step in the proof of the absence of vortices is a restriction of the GP energy to a subdomain of $ \ba $ and its splitting in a suitable energy functional plus the giant vortex profile energy. More precisely we consider the annulus $ \ann $ defined in \eqref{annulus at} with an inner radius $ \rt = \rtf - \eps^{8/7} $ suitably chosen in such a way that outside $ \ann $ the estimates \eqref{GPm exp small} and \eqref{exp small hgvm} yield the exponential smallness in $ \eps $ of both $ \gpm $ and the density profile $ \hgvm $.

We also recall the functional $ \hgvf $ introduced in \eqref{hgvf}, which is going to give the energy of the giant vortex profile, and the reduced energy
\beq
	\label{ef}
	\E_{\omega}[v] : = \int_{\ann} \diff \rv \: \hgvm^2 \lf\{|\nabla v|^2 - 2 \rmagnp \cdot (iv,\nabla v) + \eps^{-2} \hgvm^2 (1-|v|^2)^2 \ri\},
\eeq
where
\beq
	\label{current}
	 (iv,\nabla v) : = \half i \lf(  v \nabla v^* - v^* \nabla v \ri).
\eeq

	\begin{pro}[\textbf{Reduction to an annulus}]
		\label{reduction}
		\mbox{}	\\
		For any $ \omega \in \Z $ such that $ |\omega|\leq \OO(\eps^{-5/4}|\log\eps|^{-3/4}) $ and for $ \eps $ sufficiently small
		\beq
			\label{reduction est} 
			\hgve + \E_{\omega}[u_{\omega}] - \OO(\eps^{\infty}) \leq \gpe \leq \hgve + \OO(\eps^{\infty}),
		\eeq
		where the function $ u_{\omega} $ is defined in $\ann$ by the decomposition
		\beq
			\label{gvfunction u}
			\gpm(\rv) = : \hgvm(r) u_{\omega}(\rv) \exp\lf\{ i([\Omega] - \omega) \vartheta\ri\}.
		\eeq
	\end{pro}

	\begin{proof}
		As in \cite[Proposition 5.4]{CRY} the only ingredients for the proof of the above result are the exponential smallness \eqref{GPm exp small} of $ \gpm $ outside $ \ann $ and the variational equation solved by $ \hgvm $. Note that the function $ u_{\omega} $ is well defined away from the boundary $ \partial \ba $ where both $ \gpm $ and $ \hgvm $ vanish.
	\end{proof}

The idea behind the decomposition \eqref{gvfunction u} is that, if the phase factor $ \omega $ is chosen in a suitable way, the function $ u_{\omega} $ obtained by the extraction from $ \gpm $ of a density $ \hgpm $ and the giant vortex phase, i.e., the phase factor $ \exp \{ i([\Omega] - \omega) \vartheta \}$, contains basically no more vorticity and $ |u_{\omega}| \sim 1 $ in some region close to the boundary of the trap. The optimal giant vortex phase is determined by inspecting the dependence on $ \omega $ of the energy $ \hgve $, i.e., one needs to identify the $\om_0$ minimizing $\hgve$.

	\begin{pro}[\textbf{Properties of the optimal phase $ \omega_0 $ and density $ \hgvmo $}]
		\label{optimal phase pro}
		\mbox{}	\\
		For every $ \eps > 0 $ there exists an $ \omega_0 \in \Z $ minimizing $ \hgve $. Moreover one has  
		\beq
			\label{est omega_0}
			\omega_0 = \frac{2}{3 \sqrt{\pi} \eps} (1 + \OO(|\log\eps|^{-4}),	\hspace{1,5cm}	\int_{\ann} \diff \rv \: \hgvmo^2 \lf( \Omega - \frac{[\Omega] - \omega_0}{r^2} \ri) =  \OO(1).
		\eeq
	\end{pro}

	\begin{proof}
		The existence of a minimizing $ \omega_0 \in \Z $ can be deduced as in \cite[Proposition 3.2]{CRY} as well as the second estimate in \eqref{est omega_0}.
		
		The estimate of $ \omega_0 $ is a straightforward consequence of the estimates\footnote{Note that the second estimate in \eqref{est omega_0} allows to extract the simple bound $ |\omega_0| \leq \OO(\eps^{-1}) $ which guarantees that all the estimates proven in Section \ref{giant vortex densities sec} apply to $ \hgvmo $.} on $ \hgvmo $ contained in Proposition \ref{estimates gv dens}, since one has (recall the definition of the annulus $ \At $ in \eqref{annulus})
		\bml{
 			\label{est phase 1}
 			\Omega \int_{\ann} \diff \rv \lf( r^{-2} - 1 \ri) \hgvmo^2 \leq \Omega \lf(1 + \OO(\eps^{1/4} |\log\eps|^{7/4}) \ri) \int_{\At} \diff \rv \lf( r^{-2} - 1 \ri) \tfm + \OO(\eps^{-1}) \int_{\ann \setminus \At} \diff \rv \hgvmo^2 \leq	\\
			\Omega \lf(1 + \OO(\eps^{1/4} |\log\eps|^{7/4}) \ri) \int_{\tfd} \diff \rv \lf( r^{-2} - 1 \ri) \tfm + \OO(\eps^{-1}|\log\eps|^{-4}),
		}
		where we have used the fact that $ |\A\setminus\At| \leq \OO(\eps|\log\eps|^{-1}) $ and the estimates \eqref{preliminary est hgvm 2} and \eqref{pointwise bound hgvm}, which also imply that
		\bdm
			\sup_{\rv \in \A \setminus \At} \hgvmo^2(r) \leq \OO(\eps^{-1}|\log\eps|^{-3}).
		\edm
		On the other hand since
		\beq
			\label{est phase 2}
			\Omega \int_{\tfd} \diff \rv \lf( r^{-2} - 1 \ri) \tfm = \frac{\pi\eps^2\Omega^3}{4} \lf[ 1 - \rtf^4 + 2 \rtf^2 \log \rtf^{-2} \ri] = \frac{2}{3\sqrt{\pi}\eps} (1 + \OO(\eps|\log\eps|)),
		\eeq
		and
		\bdm
			 \int_{\ann} \diff \rv \: r^{-2} \hgvmo^2 \geq \rtf^{-2} (1 - \OO(\eps^{8/7})) \geq 1 - \OO(\eps|\log\eps|),
		\edm
		the result easily follows.
	\end{proof}

The analogue in the whole ball $ \B $ is discussed in the following

	\begin{pro}[\textbf{Optimal phase $ \oopt $}]
		\label{optimal phase ball pro}
		\mbox{}	\\
		For every $ \eps > 0 $ there exists an $ \oopt \in \N $ fulfilling
		\beq
			\label{est oopt}
			\oopt = \frac{2}{3 \sqrt{\pi} \eps} (1 + \OO(|\log\eps|^{-4})
		\eeq
		which minimizes $ \gve $, i.e.,
		\beq
			\giante = \gveopt.
		\eeq
	\end{pro}

	\begin{proof}
		The existence of $ \oopt $ can be proven as in Proposition \ref{optimal phase pro} above. Moreover, as in \cite[Proposition 3.2]{CRY}, it is not difficult to show that the following estimate
		\beq
			\label{compatibility}
			\int_{\B} \diff \rv \: \gvmopt^2 \lf( \Omega - \frac{[\Omega] - \omega_0}{r^2} \ri) =  \OO(1),
		\eeq
		holds true, where $ \gvmopt $ is the minimizing density associated with $ \oopt $. 
		\newline
		In order to extract the same information as in the proof of Proposition \ref{optimal phase pro} one needs however to restrict the above integration to a domain comparable to $ \ann $ and this requires some further analysis of the properties of $ \gvmopt $.
		\newline
		Using a regularization of $ \hgvmo $ as a trial function for $ \gvfopt $ and exploiting the exponential smallness \eqref{exp small hgvm} one can easily show that 
		\beq
			\label{giante est}
			\giante = \gveopt \leq \hgveo + \OO(\eps^{\infty}),
		\eeq
		 which guarantees that all the estimates proven in Proposition \ref{estimates gv dens} apply also to $ \gvmopt $. Hence one can use the exponential smallness of $ \gvmopt $ (see \eqref{exp small hgvm}) to estimate the integral inside $ \ba \setminus \ann $, but this is not completely sufficient because the potential $ \vec{B}_{\oopt} $ contains a singular term at the origin $ \sim r^{-2} $ and one needs an additional estimate showing that $ \gvmopt $ vanishes as $ r \to 0 $. This is proven in Lemma \ref{pointwise origin} below.
		\newline
		By using \eqref{point est origin} and the analogue of \eqref{exp small hgvm}, one thus obtains from \eqref{compatibility}
		\beq
			\int_{\ann} \diff \rv \: \gvmopt^2 \lf( \Omega - \frac{[\Omega] - \omega_0}{r^2} \ri) =  \OO(1),
		\eeq
		which implies the result exactly as\footnote{Note that the other estimates of $ \gvmopt $ (analogous to those stated in Proposition \ref{estimates gv dens}) which are  needed to complete the proof can be derived from the energy bound \eqref{giante est}.} in the proof of Proposition \ref{optimal phase pro}.
	\end{proof}

	\begin{lem}[\textbf{Pointwise estimate of $ \gvmopt $ close to the origin}]
		\label{pointwise origin}
		\mbox{}	\\
		The density $ \gvmopt $ minimizing the functional $ \gvfopt $ defined in \eqref{gvf} satisfies the pointwise estimate
		\beq
			\label{point est origin}
			\gvmopt(r) \leq  \lf\| \gvmopt \ri\|_{L^{\infty}(\B)} (2r)^{[\Omega/2]}.
		\eeq
		for any $ 0 \leq r \leq 1/2 $.
	\end{lem}

	\begin{proof}
		The function $ W(r) : = \lf\| \gvmopt \ri\|_{\infty} (2r)^{[\Omega/2]} $ is a supersolution in $ [0,1/2] $ for the variational equation solved by $ \gvmopt $, i.e.,
		\bdm
			- \Delta \gvmopt + ([\Omega] - \oopt)^2 r^{-2} \gvmopt - 2 \Omega ([\Omega] - \oopt) \gvmopt + 2 \eps^{-2} \gvmopt^3 = \mu_{\mathrm{opt}} \gvmopt,
		\edm
		since
		\bmln{
 			- \Delta W + ([\Omega] - \oopt)^2 r^{-2} W - 2 \Omega ([\Omega] - \oopt) W + 2 \eps^{-2} W^3 - \mu_{\mathrm{opt}} W \geq	\\
			\lf\{ \lf[  ([\Omega] - \oopt)^2 - [\Omega/2]^2 - C \Omega \ri]  r^{-2} - 2 \Omega ([\Omega] - \oopt) -  \mu_{\mathrm{opt}} \ri\} W(r) \geq C \Omega W(r) \geq 0,
		}
		where we have used the estimate $  \mu_{\mathrm{opt}}  = -\Omega^2 (1 - o(1)) $ and the fact that we are in the interval $ r \in [0,1/2]$.
		\newline
		Since at the boundary $ \partial \ba_{1/2} $ one has $ \gvmopt(1/2) \leq  \lf\| \gvmopt \ri\|_{\infty} = W(1/2) $, the maximum principle (see, e.g., \cite{Evans}) guarantees that $ \gvmopt(r) \leq W(r) $ and therefore the result.
	\end{proof}

\subsection{Estimates of the Reduced Energies}

The next crucial step in the proof of the absence of vortices is the lower bound for the reduced energy functional $ \E_{\omega_0} $ and in the rest of this section we will focus on such a problem. Since the optimal phase $ \omega_0 $ as well as the associated density $ \hgvmo $ can be fixed throughout the rest of the proof, we simplify the notation for the sake of clarity and set
\beq
	\label{redefinitions}
	\E_{\omega_0}[v] =: \E[v],	\qquad 	\hgvmo = : g, 	\qquad	\rmaxgv(\omega_0) = : \rmaxgv, 	\qquad	\vec{B}_{\omega_0}(r) = : \vec{B}(r) = \lf[  \Omega r - \lf( [\Omega] - \omega_0 \ri) r^{-1} \ri] \vec{e}_{\vartheta},
\eeq
and
\beq
	\label{F functional}	
	\F [v] := \int_{\ann} \diff \rv \: g ^{2}  \lf\{ \left| \nabla v \right|^2 + \eps^{-2} g^2 \left(1-|v|^2 \right)^2 \ri\},
\eeq
where $ \rt : = \rtf - \eps^{8/7} $ and (see \eqref{annulus at})
\bdm
	\ann : =  \lf\{ \rv \in \ba : \:  r \geq \rt \ri\}.
\edm
We also recall that $ u :=u_{\om_0} $ is defined inside $ \ann $ by
\bdm
	\gpm(\rv) = : g(r) u(\rv) \exp\lf\{ i([\Omega] - \omega_0) \vartheta\ri\},
\edm
and the annulus $ \At $ is (see \eqref{annulus})
\bdm
 	\At :=  \lf\{ \rv \in \ba : \: \rb \leq r \leq 1 - \eps^{3/2} |\log\eps|^2 \ri\}
\edm
with $ \rb := \rtf + \ep |\log \ep|^{-1} $ (see \eqref{rd gv}). Note that thanks to the pointwise estimate \eqref{pointwise bound hgvm}, we have the lower bound
\begin{equation}
	\label{g lower bound}
	 g^2(r) \geq \frac{C}{\ep |\log \ep|^{3}} \mbox{ on } \At.
\end{equation}

We can now state the main result in this section, which is going to be the crucial ingredient in the proof of the absence of vortices:

	\begin{pro}[\textbf{Bounds on the reduced energies}]
		\label{reduced en bounds pro}
		\mbox{}	\\
		If $ \Omega = \Omega_0 \eps^{-2} |\log\eps|^{-1} $ with $\Om_{0}> 2(3 \pi)^{-1} $, then for $\eps$ small enough
		\beq
			\label{reduced en bounds}
			\F[u] \leq \OO\lf( \frac{|\log\eps|^2}{\log|\log\eps|^2}\ri),	\hspace{1,5cm}	\E [u] \geq - \OO\lf( \frac{|\log\eps|^2}{\log|\log\eps|^2}\ri).
		\eeq
	\end{pro}

The proof of the above results is quite involved and before the discussion of its details, which is postponed to Section \ref{gv main proofs}, we are going to give a quick sketch of it together with the statement of several preliminary results. 

The main trick in the estimate of the reduced energy is an integration by parts of the second term in \eqref{ef}, which is made possible by the introduction of a potential function $ F(r) $ already considered in \cite{CRY}. Such a function satisfies the key properties
\begin{equation}
	\label{properties F}
	\nabla ^{\perp} F  = 2 g ^2 \vec{B}, 	\hspace{1,5cm}	 F(R_{<}) = 0,
\end{equation}  
and it is explicitly given by
\begin{equation}
	\label{F}
	F(r):= 2 \int_{R_{<}}^r \diff s \: g^2  (s) \left(\Om  s- \left([\Om] - \om_0 \right)\frac{1}{s}\right) = 2 \int_{R_{<}}^r \diff s \: g ^2(s) \vec{B}(s) \cdot \vec{e}_{\vartheta}.
\end{equation}
Other important properties of $ F $ are formulated in the next lemma and are basically straightforward consequences of \eqref{est omega_0} and the bound
\beq
	\label{mangp inside}
	|B(r)| \leq \OO(\eps^{-1}) \:\: \mbox{on } \ann,
\eeq
which follows from the definition of $ \ann $.
	
	\begin{lem}[\textbf{Useful properties of  $F$}]
		\label{properties F lem}
		\mbox{}\\
		Let $F$ be defined in \eqref{F}. The following bounds hold true:
		\beq
			\label{F bounds 1}
			\lf\| F \ri\|_{L^{\infty} (\ann)} \leq \OO(\ep ^{-1}),	\hspace{1,5cm}	\lf\| \nabla F \ri\|_{L^{\infty} (\ann)}  \leq \OO(\ep ^{-2} |\log \ep|^{-1}).					
		\eeq
		Moreover one has the pointwise estimates
		\begin{equation}
			\label{F bounds 2}
			\lf| F(1) \ri| \leq \OO(1),		\hspace{1,5cm}	|F(r)| \leq C 
			\begin{cases}
				\eps^{-1} |r - \rt| g^2 (r),	&	\mbox{if } r \in [\rt,\rmaxgv],	\\
				1 + \ep^{-1} |1 - r| g^2(r) ,	&	\mbox{if } r \in [\rmaxgv,1].
			\end{cases}
		\end{equation}
	\end{lem}

	\begin{proof}	
		Most of the proof follows from \cite[Lemma 4.1]{CRY}. The estimate \eqref{est omega_0} yields $ |F(1)| \leq \OO(1) $. The last inequality in \eqref{F bounds 2} for $ r \in [\rmaxgv,1] $ is a consequence of this bound together with the identity
		\bdm
			F(r) = F(1) - 2 \int_{r}^1 \diff s \: g ^2(s) \vec{B}(s) \cdot \vec{e}_{\vartheta},
		\edm
		and the fact that $ g(r) $ is decreasing for $ r \in [\rmaxgv,1] $.
	\end{proof}

Due to the lack of control of the behavior of the function $ u $ at the boundary $ \partial \ba $, we need to use a suitable decomposition of $ F$: An integration by parts (Stokes theorem) of the second term in \eqref{ef} would indeed give
\begin{equation}
	\label{integration by parts 1}
	- 2 \int_{\ann} \diff \rv \: g ^2 \vec{B} \cdot (iu,\nabla u) = \int_{\ann} \diff \rv \: F(r) \curl (iu,\nabla u)  - \int_{\partial \B} \diff \sigma \: F(1) (iu , \partial_{\tau}u),
\end{equation} 
and the last term in the expression above clearly depends on $ u $ at the boundary. While Neumann boundary conditions allow to extract some information about $ u $ on $ \partial \ba $ and in particular an upper estimate for that term, on the opposite, if Dirichlet conditions are imposed, $ u $ is not even well posed on $ \partial \ba $, since both $ \gpm $ and $ g$ vanish there. A way out to avoid such a problem is the decomposition of $ F $ into a function vanishing on $ \partial \ba $ and another one whose gradient can be explicitly controlled: More precisely we set
\beq
	\label{Fout}
	\fout(r) : = F(1) \bigg[ \int_{\rt}^1 \diff s \: s^{-1} g^2(s) \bigg]^{-1} \int_{\rt}^r \diff s \: s^{-1} g^2(s),
\eeq
so that 
\beq
	\label{Fout properties}
	\nabla \lf( g^{-2} \nabla \fout \ri) = 0,	\hspace{1,5cm}	\fout(1) = F(1).
\eeq
If we now define
\beq
	\label{Fin}
	\fin(r) : = F(r) - \fout(r),
\eeq
one can easily verify that 
\beq
	\label{Fin properties}
	\nabla \lf( g^{-2} \nabla \fin \ri) = 2 \nabla \cdot B(r) \vec{e}_r,	\hspace{1,5cm}	\fin(1) = 0,
\eeq
and, integrating by parts only the term involving $ \fin $ in \eqref{integration by parts 1} we obtain
\begin{equation}
	\label{integration by parts 2}
	- 2 \int_{\ann} \diff \rv \: g ^2 \vec{B} \cdot (iu,\nabla u) = - \int_{\ann} \diff \rv \: \nabla^{\perp} \fout \cdot (iu,\nabla u) + \int_{\ann} \diff \rv \: \fin(r) \curl (iu,\nabla u).
\end{equation} 
The energy $ \E[u] $ can thus be rewritten as
\begin{equation}
	\label{integ by parts energy}
	\E [u] = \int_{\ann} \diff \rv \lf\{ g ^{2} \left| \nabla u \right|^2 + \fin(r) \curl (iu,\nabla u) +  \eps^{-2} g^4 \left( 1-|u|^2 \right)^2 \ri\} - \int_{\ann} \diff \rv \: \nabla^{\perp} \fout \cdot (iu,\nabla u).
\end{equation}
The first three terms above are the most important ones and their estimate is the key result in the proof of the absence of vortices. The last term on the other hand can be estimate separately and one can show that it yields only a smaller order correction.
\newline
More precisely the first two terms can be estimated in terms of the vorticity of $u$: As in \cite{CRY}, if we suppose that $|u|\sim 1$ except in some balls $ \lf\{ \B(\avj,t) \ri\}_{j\in J}$, $ J \subset \N $, whose radius $ t $ is much smaller than the width of $\At$, and we denote by $d_j$ the degree of $u$ around $\avj$, 
\begin{equation}
	\label{vort heur 1}
	\int_{\ann} \diff \rv \: \fin(r) \curl (iu,\nabla u)\simeq \sum_{j \in J} 2 \pi \fin(a_j) d_j.    
\end{equation}
and, optimizing w.r.t. the radius $ t $,
\begin{equation}
	\label{vort heur 2}
	\int_{\ann} \diff \rv \: g ^{2} \left| \nabla u \right|^2 \gtrapprox \sum_{j \in J} 2\pi g^2 (a_j) |d_j| \log \left(\frac{\ep |\log \ep|}{t}\right)  \gtrapprox \sum_{j \in J} \pi g^2 (a_j) |d_j| |\log \ep |. 
\end{equation}
Hence
\begin{equation}
	\label{vort heur 3}
 	\E [u] \gtrapprox \sum_{j \in J} 2 \pi  |d_j| \left( \half g^2 (a_j) |\log \ep | +  \fin(a_j) \right),
\end{equation}
and, if $\Om_0> 2(3\pi) ^{-1}$, the sum between parenthesis is positive for any $\avj$ in the bulk (see Section \ref{ang vel 3 sec}), which means that vortices become energetically unfavorable. Note that there is an important difference with the analysis contained in \cite{CRY} since $ F $ is replaced in the expression above by $ \fin $. This is basically the main effect of Dirichlet boundary conditions.

The starting point of the reduced energy estimate is given by the following preliminary upper bounds:

	\begin{lem}[\textbf{Preliminary energy bounds}]
		\label{initial bounds lem}
		\mbox{}	\\
		If $ \Omega \sim \eps^{-2} |\log\eps|^{-1} $ as $ \eps \to 0  $,
		\beq
			\label{F and E first bounds}
			\F [u] \leq \OO(\ep ^{-2}), 	\hspace{1,5cm}	\E [u] \geq - \OO(\ep ^{-2}).
		\eeq
	\end{lem}

	\begin{proof}
		See \cite[Lemma 4.2]{CRY}.
	\end{proof}

\subsection{Vortex Ball Construction and Jacobian Estimate}
\label{vortex ball construction sec}

In order to construct families of balls containing all the vortices of $ u $, we need to exploit some local energy bound on $ \F[u] $. However the bounds \eqref{F and E first bounds} are not sufficient for our purposes, since they imply that the area of the set where $u$ can possibly vanish is of order $\ep^2 |\log \ep|^2$, whereas the vortex balls method requires to cover it by balls whose radii are much smaller than the width of $\ann$, which is $\OO(\ep |\log \ep|)$. 
\newline
As in \cite{CRY} there is a way out to this obstruction in the localization of the energy bound \eqref{F and E first bounds}, given by the decomposition of the domain into suitable good and bad cells:

	\begin{defi}[\textbf{Good and bad cells}]
		\mbox{}	\\
		We decompose $ \ann $ into almost rectangular cells $ \A_n $, $ n \in \N $, of side length $ \OO( \ep |\log \ep| )$, given by 
		\beq
			\A_n : = \lf\{ \rv \in \ann : \: \vartheta \in [n \theta, (n+1) \theta [ \ri\},
		\eeq
		where $ \theta : = 2\pi / N $ and $ N \sim \ep^{-1} |\log \ep|^{-1} $ is the total number of cells. Let $ 0 \leq \alpha < \frac{1}{2}$ be a parameter to be fixed later on. 
		\newline
		We say that $ \A_n $ is an $\alpha$-good cell if 
		\begin{equation}
			\label{def good cells}
			\int_{\A_n}  \diff \rv \: g ^{2} \lf\{ \left| \nabla u \right|^2+ \eps^{-2} g ^2 \left(1-|u|^2 \right)^2 \ri\} \leq \ep^{-1-\alpha}|\log \ep|,
		\end{equation}
		whereas inside $ \alpha$-bad cells the (strict) inequality is reversed. We denote by $ N_{\alpha} ^{\mathrm{G}} $ and $ N_{\alpha} ^{\mathrm{B}} $ the numbers of $\alpha$-good and bad cells and by $ GS_{\alpha} $ and $ BS_{\alpha} $ the sets covered by good and bad cells respectively.
	\end{defi}

By definition of bad cells, one has that \eqref{F and E first bounds} immediately implies
	\begin{equation}\label{number bad}
		N_{\alpha} ^{\mathrm{B}} < \ep^{1+\alpha} |\log \ep|^{-1} \F [u] \leq C \ep^{-1+\alpha} |\log \ep|^{-1} \ll N, 	 
	\end{equation}
i.e., there are very few $\alpha$-bad cells. Note also that the final estimate \eqref{reduced en bounds} implies that there are actually no bad cells at all.

We can now construct the vortex balls inside good cells but, since the density has to be large enough, we need to restrict the analysis to the subdomain $ \At \subset \ann $ (see \eqref{annulus} for its definition):

	\begin{pro}[\textbf{Vortex ball construction inside good cells}]
		\label{vortex balls pro}
 		\mbox{}\\
		For any $0\leq \al< \frac{1}{2}$ and $ \eps $ small enough, there exists a finite collection $ \{ \B_i \}_{i \in I} := \left\lbrace \B (\avi, \varrho_i)\right\rbrace_{i\in I}$ of disjoint balls with centers $ \avi $ and radii $ \varrho_i $ such that
		\begin{enumerate}
			\item $\left\lbrace \rv \in GS_{\al} \cap \At : \: \left| |u| - 1  \right| > |\log \ep| ^{-1}  \right\rbrace \subset \bigcup_{i \in I} \B_i$,
			\item for any $\al $-good cell $\A_n$, $\sum_{i, \: \B_i \cap \A_n \neq \varnothing } \varrho_i = \ep |\log \ep |^{-5} $.		   
		\end{enumerate}
		Setting $d_i:= \dg \{ u, \partial \B_i \} $, if $ \B_i \subset \At \cap GS_{\al} $, and $d_i=0$ otherwise, we have the lower bounds
		\begin{equation}\label{lowboundballs}
		 	\int_{\B_i}  \diff \rv \: g ^2 \left|\nabla u\right|^2 \geq 2\pi \left(\frac{1}{2} -\al \right) |d_i|   g^2 (a_i) \left| \log \ep \right| \left(1-C \frac{\log \left| \log \ep \right|}{\left|\log \ep\right|}\right).
		\end{equation}
	\end{pro}

	\begin{proof}
		See \cite[Proposition 4.2]{CRY}.
	\end{proof}

Given a suitable family of disjoint balls as in the above proposition, one can prove that in the $\al$-good set the vorticity measure of $u$ will be close to a sum of Dirac masses, i.e.,
	\[
		\curl(iu,\nabla u)\simeq \sum_{i \in I} 2\pi d_i \delta(\rv - \avi),
	\]
where $\delta(\rv - \avi)$ stands for the Dirac delta centered at $ \avi $.

	\begin{pro}[\textbf{Jacobian estimate}]
		\label{jacobian estimate pro}
 		\mbox{}\\
		Let $0\leq \al<\frac{1}{2}$ and $\phi$ be any piecewise-$C^1$ test function with compact support $ {\rm supp}(\phi) \subset \At \cap GS_{\al} $. Let also $\left\lbrace \B_i \right\rbrace_{i\in I} : = \lf\{ \B(\avi,\varrho_i) \ri\}_{i \in I} $ be a disjoint collection of balls as in Proposition \ref{vortex balls pro}.
		\newline
		Then setting $d_i:= \dg \{ u, \partial \B_i \} $, if $ \B_i \subset \At \cap GS_{\al} $, and $d_i=0$ otherwise, one has
		\begin{equation}
			\label{jacobian estimate}
			\bigg|\sum_{i\in I} 2\pi d_i \phi (\avi)- \int_{GS_{\al} \cap \At} \diff \rv \: \phi \:  \curl (iu,\nabla u) \bigg| \leq  C \left\Vert \nabla \phi \right\Vert_{L^{\infty}(GS_{\al})} \ep ^{2} |\log \ep|^{-2} \F[u].  
		\end{equation}	
	\end{pro}

	\begin{proof}
		See \cite[Proposition 4.3]{CRY}.
	\end{proof}

\subsection{Completion of the Proofs}
\label{gv main proofs}

The main goal in this section is the proof of Proposition \ref{reduced en bounds pro}, which will lead to the proof of Theorem \ref{giant vortex teo}.
\newline
As anticipated before, the first important step is an integration by parts of the second term in \eqref{ef}, but, since it has to be performed cell by cell, it generates boundary terms living on the frontiers between good and bad cells. Such terms are artificial, since the cell decomposition has no physical meaning, and we want to avoid having to estimate them. 
\newline
As in \cite{CRY} we introduce an azimuthal partition of unity to get rid of these terms (see also \cite[Definition 4.2 and Eq. (4.69)]{CRY}): We define a pleasant set $PS_{\al}$ as the set generated by good cells such that their neighbor cells are both good (pleasant cells), whereas the average set $ AS_{\alpha} $ is made of good  cells with exactly one good cell as neighbor (average cells). Finally the unpleasant set $ UPS_{\al} $ contains all the remaining good and bad cells (unpleasant cells). Denoting by $N_{\al}^{\mathrm{P}}$, $N_{\al}^{\mathrm{AS}} $, $N_{\al}^{\mathrm{UP}}$ the number of pleasant, average and unpleasant cells respectively, it is not difficult to see that
\begin{equation}
	\label{number pleasant}
	N_{\al} ^{\mathrm{UP}} \leq \thalf N_{\al} ^{\mathrm{B}} \ll N,	\hspace{1,5cm}	N_{\al} ^{\mathrm{A}} \leq 2 N_{\al} ^{\mathrm{B}} \ll N.
\end{equation}
The partition of unity is given by two functions $ \chin(\vartheta)$ and $\chout(\vartheta)$ such that $ \chin(\vartheta)  +  \chout(\vartheta) = 1 $ for any $ \vartheta \in [0,2\pi] $ and
\beq
	\label{angular partition}
	\chout(\vartheta) := 
	\begin{cases}
		1, 	& 	\mbox{ if } \vartheta \in UPS_{\al},	\\
		0,  	&	\mbox{ if } \vartheta \in PS_{\al},
	\end{cases}
	\hspace{1cm}
	\chin(\vartheta)  :=
	\begin{cases}
		1, 	& 	\mbox{ if } \vartheta \in PS_{\al},	\\
		0, 	&	\mbox{ if } \vartheta \in UPS_{\al}, 
	\end{cases}
\eeq
Since both functions vary from $ 1 $ to $ 0 $ inside an average cell, one can always impose the bounds
\beq
	\label{gradient chi}
 	\lf|\nabla \chout \ri| \leq \OO(\eps^{-1}|\log\eps|^{-1}),	\hspace{1,5cm}	\lf|\nabla \chin \ri| \leq \OO(\eps^{-1}|\log\eps|^{-1}),
\eeq
because the side length of a cell is $\propto \ep |\log \ep|$.

In order to apply the jacobian estimate proven in Proposition \ref{jacobian estimate pro} to the function $ \phi = \chin \fin $, whose support is not contained in $ \At $ but only in $ \ann $, we also need a radial partition of unity: We define two radii as (recall that $ \rb = \rtf + \eps |\log\eps|^{-1} $ as in \eqref{rd gv})
\beq
	\label{R cut}
	R_{\mathrm{cut}}^+ := 1 -\ep |\log \ep|^{-1},	\hspace{1,5cm}	R_{\mathrm{cut}}^- : = \rb + \ep |\log \ep|^{-1},
\eeq
and two positive functions $ \xiin(r) $ and $ \xiout(r) $  satisfying $ \xiin(r) + \xiout(r) = 1 $ for any $ \rv \in \ann $ and (recall \eqref{rt}, i.e., $ \rt = \rtf - \eps^{7/6} $)
\begin{eqnarray}
	\xiout(r) &:=&
	\begin{cases}
		1, 	& 	\mbox{ if } r \in [\rt, \rb],	\mbox{ or } [1 - \eps^{3/2} |\log\eps|^2, 1] \\
		0,  	&	\mbox{ if } r \in [R_{\mathrm{cut}}^-, R_{\mathrm{cut}}^+],
	\end{cases}		\\
	\label{radial partition}
	\xiin(r) &:=&
	\begin{cases}
		1, 	& 	\mbox{ if } r \in [R_{\mathrm{cut}}^-, R_{\mathrm{cut}}^+],	\\
		0, 	&	\mbox{ if } r \in [\rt, \rb],	\mbox{ or } [1 - \eps^{3/2} |\log\eps|^2, 1].  
	\end{cases}
\end{eqnarray}
Thanks to \eqref{R cut} we can also assume
\beq
	\label{gradient xi}
 	\lf|\nabla \xiout \ri| \leq \OO(\eps^{-1}|\log\eps|),	\hspace{1,5cm}	\lf|\nabla \xiin \ri| \leq \OO(\eps^{-1}|\log\eps|).
\eeq

We are now ready to prove the bound on the reduced energies:

	\begin{proof}[Proof of Proposition \ref{reduced en bounds pro}]
		\mbox{} \\
		For the sake of simplicity we denote by $ \lf\{ \B_i \ri\}_{i \in I} : = \left\lbrace \B(\avi, \varrho_i)\right\rbrace_{i\in I}$ a collection of disjoint balls as in Proposition \ref{vortex balls pro}, whereas the subset $ J \subset I $ identifies balls such that $ d_j \neq 0 $. 

		The starting point is an integration by parts as in \eqref{integ by parts energy}, i.e.,
		\beq
			\label{part int}
			\E [u] = \int_{\ann} \diff \rv \lf\{ g ^{2} \left| \nabla u \right|^2 + \fin(r) \curl (iu,\nabla u) +  \eps^{-2} g^4 \left( 1-|u|^2 \right)^2 \ri\} - \int_{\ann} \diff \rv \: \nabla^{\perp} \fout \cdot (iu,\nabla u).
		\eeq
		
		The last term in the expression above is the easiest to bound: By using the explicit expression of $ \fout $, one obtains
		\bml{
 			\label{fout bound}
 			\lf| \int_{\ann} \diff \rv \: \nabla^{\perp} \fout \cdot (iu,\nabla u) \ri| \leq C \int_{\ann} \diff \rv \: g^2(r) |u| \lf| \nabla u \ri| \leq C \lf( \delta \int_{\ann} \diff \rv g^2 |u|^2 + \delta^{-1} \int_{\ann} \diff \rv g^2 \lf|\nabla u \ri|^2 \ri) \leq		\\
			C \lf(\delta + \delta^{-1} \F[u] \ri) \leq C \F[u]^{1/2},
			}
		where we have introduced a parameter $ \delta $ and chosen $ \delta = \F[u]^{1/2} $ (recall that $ \F[u] \geq 0 $).

		The remaining term in \eqref{part int} can be estimated exactly as in \cite[Proof of Proposition 4.1]{CRY}, with only one difference due to the presence of $ \fin $ instead of $ F $: Since by definition the former vanishes on $ \partial \ba $, we can get rid of all the boundary terms (see, e.g., \cite[Eq. (4.86)]{CRY}) and the final result is, for some parameters $\gamma, \delta$ that we fix below, 
		\begin{multline}
			\label{lowbound1}
			\int_{\ann} \diff \rv \lf\{ g ^{2} \left| \nabla u \right|^2 +  \fin \: \curl (iu,\nabla u) \ri\} \geq 
			\\  \sum_{j\in J} \xi_{\mathrm{in}} (a_j) |d_j| \left[ \left(1-\gamma \right)\left(\frac{1}{2} -\al \right)  g^2 (a_j) \left| \log \ep \right| \left(1-C \frac{\log \left| \log \ep \right|}{\left|\log \ep\right|}\right) -  |\fin(a_j)|   \right] +
			\\ \left(1-\gamma \right) \int_{\ann} \diff \rv \: \xi_{\mathrm{out}} g^2 |\nabla u |^2 -  \int_{\ann} \diff \rv \: \xi_{\mathrm{out}}  |\fin(r)| |\nabla u|^2 + \left(\gamma -\delta \right) \int_{\ann} \diff \rv \: g^2 | \nabla u | ^2
			\\  - \frac{C}{\delta \ep ^2} \int_{UPS_{\al}\cup AS_{\al}} \diff \rv \: g^2 |u|^2  - C |\log \ep|^{-1} \Fg[u].
		\end{multline}
		We can now choose the parameters $ \alpha $, $ \delta $ and $ \gamma $  as follows:
		\begin{equation}\label{parameters}
		 	\gamma = 2\delta = \frac{\log |\log \ep|}{|\log \ep|}, \hspace{1,5cm} \al = \alt \frac{\log |\log \ep|}{|\log \ep|},
		\end{equation}
		where $\alt$ is a large enough constant (see below). 

		Using the properties of the function $ H(r) : = \half g^2 \left| \log \ep \right| - |\fin| $ proven in Proposition \ref{ang vel 3 pro}, we have for $ \Omega_0 > (3\pi)^{-1} $
		\[
		 	\half g^2 (a_j)\left| \log \ep \right| - |\fin(a_j)| \geq C \eps^{-1} |\log\eps|^{-2}
		\]
		for any $ \avj \in \At $, so that
		\begin{multline}\label{lowbound11}
			\left(1-\gamma \right)\left(\half -\al \right)  g^2 (a_j) \left| \log \ep \right| \left(1-C \frac{\log \left| \log \ep \right|}{ \left|\log \ep\right|}\right) -  |\fin(a_j)|   \geq  \\  \half g^2 (a_j)\left| \log \ep \right| - |\fin(a_j)| - C g^2 (a_j) \log \left| \log \ep \right| \geq C \ep^{-1} |\log\eps|^{-2} \lf( 1 - \frac{C\log \left| \log \ep \right|}{\lf|\log \ep \ri|} \ri) > 0,
		\end{multline}
		where we have used \eqref{g lower bound}.	
		\newline
		On the other hand for any $ \rv \in \supp(\xiout) $ either $ |r-\rt| \leq C \ep |\log \ep|^{-1} $ or $ |r -1 | \leq C \ep |\log \ep|^{-1} $, which by the bounds \eqref{F bounds 2} imply that in the first case
		\beq
		 	\lf|\fin(r) \ri| \leq C \lf( |\log\eps|^{-1} g^2(r) + \lf| \fout(r) \ri| \ri) \leq C \lf( |\log\eps|^{-1} g^2(r) + 1 \ri) \leq C |\log\eps|^{-1} g^2(r),
		\eeq
		thanks to \eqref{g lower bound}, whereas in the second case
		\bml{
		 	\lf| \fin(r) \ri| \leq \lf| F(1) - \fout(r) \ri| + 2 \int_{r}^1 \diff s \: |B(s)| g^2(s) \leq		\\
			|F(1)| \lf(\int_{\rt}^1 \diff s \: s^{-1} g^2(s) \ri)^{-1} \int_{r}^1 \diff s \: s^{-1} g^2(s) + C |\log\eps|^{-1} g^2(r) \leq C |\log\eps|^{-1} g^2(r).
		}
		Note that in this second case there is no need to assume that $ r \geq \rmax $ in order to use that $ g $ is decreasing in $ [r,1] $: By the bounds \eqref{preliminary est hgvm 2} and \eqref{pointwise bound hgvm}, for any $ 1 - \eps |\log\eps|^{-1} \leq r \leq 1 $,
		\bdm
			 g^2(r) \geq (1 - o(1)) \tfm(1 - \eps |\log\eps|^{-1}) \geq (1 - o(1)) g^2(\rmaxgv),
		\edm
		so that we can always bound in the integrals $ g^2(s) $ by $ (1+o(1)) g^2(r) $.
		\newline
		In conclusion
		\beq
			\lf| \fin(r) \ri| \ll g^2(r),
		\eeq
		for any $ \rv \in  \supp(\xiout) $ and thus
		\begin{equation}\label{boundarylayer}
			\int_{\ann} \diff \rv \: \xi_{\mathrm{out}} \lf[ \left(1-\gamma \right) g^2(r) - \lf| \fin(r) \ri| \ri] |\nabla u |^2  \geq 0.
		\end{equation}

		Finally we have from \eqref{lowbound1}, \eqref{lowbound11} and \eqref{boundarylayer}
		\begin{multline}
			\label{lower bound 2}
			\int_{\ann} \diff \rv \lf\{ g ^{2} \left| \nabla u \right|^2 + \fin(r) \: \curl (iu,\nabla u) \ri\} \geq
			\\ C \frac{\log \left| \log \ep \right|}{\left|\log \ep\right|} \int_{\ann} \diff \rv \: g^2 | \nabla u | ^2 -C  \frac{|\log \ep|}{ \ep ^2 \log \left| \log \ep \right|} \int_{UPS_{\al}\cup AS_{\al}} \diff \rv \: g^2 |u|^2  - C | \log \ep | ^{-1} \Fg[u],
		\end{multline}
		and adding 
		\bdm
			\int_{\ann} \diff \rv \: \frac{g^4}{\ep^2}\left(1-|u|^2 \right)^2
		\edm	
		to both sides of \eqref{lower bound 2} and using \eqref{part int} and \eqref{fout bound}, we get the lower bound
		\begin{equation}\label{lower bound 3}
		 	\Eg [u] \geq C \lf\{ \frac{\log \left| \log \ep \right|}{\left|\log \ep\right|}  \Fg [u] - \F[u]^{1/2} - \frac{|\log \ep|}{ \ep ^2 \log \left| \log \ep \right|} \int_{UPS_{\al}\cup AS_{\al}} \diff \rv \: g^2 |u|^2 \ri\},
		\end{equation}
		valid for $\ep$ small enough and $\Om_0 > (3\pi)^{-1} $. But $g^2 |u|^2 = |\gpm|^2 \leq C \ep^{-1} |\log \ep|^{-1}$, whereas the side length of a cell is $\OO (\ep |\log \ep|)$, thus
		\begin{equation}\label{bad term 0}
		 	\int_{UPS_{\al}\cup AS_{\al}} \diff \rv \: g^2 |u|^2 \leq  \frac{C\left| UPS_{\al}\cup AS_{\al} \right|}{\ep |\log \ep|} \leq C \ep |\log \ep| \left( N_{\al} ^ {\mathrm{UP}} + N_{\al} ^ {\mathrm{A}} \right) \leq C \ep ^{2+\al} \F [u], 
		\end{equation}
 		by \eqref{number bad} and \eqref{number pleasant}. Therefore \eqref{lower bound 3} becomes
		\begin{equation}
			\label{lower bound 4}
			\Eg [u] \geq C \lf\{ \frac{\log \left| \log \ep \right|}{\left|\log \ep\right|}  \Fg [u] - \frac{|\log \ep|}{\log|\log \ep|}\ep ^{\al} \F[u] - \F [u]^{1/2} \right\}.
		\end{equation}
		Recalling the choice of $\al$ in \eqref{parameters}, we now take a constant $ \alt > 2$, so that
		\[
			 \frac{|\log \ep|}{\log|\log \ep|}\ep ^{\al} = \frac{|\log \ep|^{1-\alt}}{\log|\log \ep|} \ll \frac{\log \left| \log \ep \right|}{\left|\log \ep\right|}
		\]
		and 
		\begin{equation}
			\label{lower bound 5}
		 	\OO(\ep^{\infty}) \geq \Eg [u] \geq \left(\frac{\log \left| \log \ep \right|}{\left|\log \ep\right|}  \Fg [u] -  \F[u]^{1/2} \right)
		\end{equation}
		which yields both results.
	\end{proof}

The proof of the energy asymptotics is essentially a corollary of the reduced energy estimates together with the discussion contained in Section \ref{giant vortex densities sec}:

	\begin{proof}[Proof of Theorem \ref{giant vortex teo}]
		\mbox{}	\\
		By \eqref{reduction est} and \eqref{reduced en bounds}
		\[
 			\gpe \geq \hgveo + \E[u] - \OO(\eps^{\infty}) \geq \hgveo - C (\log|\log\eps|)^{-2} |\log\eps|^2,
		\]
		but one can easily show that $ \hgveo \geq \gveo - \OO(\eps^{\infty}) $ by simply testing the functional $ \gvfo $ on a suitable regularization of $ \hgvmo $ and thus
		\[
 			\gpe \geq \gveo - C (\log|\log\eps|)^{-2} |\log\eps|^2 \geq \giante -  C (\log|\log\eps|)^{-2} |\log\eps|^2,
		\]
		which concludes the lower bound proof. 
		
		The upper bound \eqref{giante} is trivially obtained by testing the GP functional on a giant vortex function with phase $ [\Omega] - \oopt $ (see Proposition \ref{optimal phase ball pro}).
	\end{proof}

Using the equations satisfied by $\gpm$ and $g$, one can derive an equation satisfied by $u$: 
\[
 -\nabla (g ^2 \nabla u) -2i g^2 \vec{B} \cdot \nabla u +2\frac{g ^4}{\ep ^2} \left( |u|^2 -1 \right)u =\lambda g ^2 u
\]
where $\lambda =  \chem - \hgvchemo $. A useful estimate on the gradient of $u$ follows from this equation and allows to conclude the proof of Theorem \ref{no vortices}. We state the estimate for convenience and refer to \cite[Lemma 5.1]{CRY} for its proof.  

\begin{lem}[\textbf{Estimate for the gradient of $u$}]
	\mbox{}	\\
	 Recall the definition of $u$ in \eqref{gvfunction u}. There is a finite constant $C$ such that
	\begin{equation}\label{gradu}
	 	\lf\| \nabla u \ri\|_{L^{\infty}(\At)} \leq C \frac{|\log \ep| ^{3/2}}{\ep ^{3/2}}.
	\end{equation}
\end{lem}

We now complete the 

	\begin{proof}[Proof of Theorem \ref{no vortices}]
Suppose that at some point $ \rv_0 \in \At$ we have 
	\[
 		\left| |u(\rv_0)| - 1 \right| \geq   \ep^{1/4} |\log \ep|^{3}.
	\]
Then, using \eqref{gradu}, there is a constant $C$ such that, for any $ \rv \in \B(\rv_0, C \ep ^{7/4} |\log \ep | ^{3/2})$, we have 
	\[
	 \left| |u(\rv)| -1 \right| \geq  \half \ep^{1/4} |\log \ep|^{3} .
	\]
This implies (recall \eqref{g lower bound})
	\[
	 \int_{\B(\rv_0, C \ep ^{7/4} |\log \ep | ^{3/2})} \diff \rv \: \frac{g^4}{\ep ^2 }\left(1-|u|^2\right) ^2 \geq C |\log \ep|^3,
	\]
and thus 
	\begin{equation}\label{Fginf}
	 	\F [u] \geq C |\log \ep|^3,
	\end{equation}
which is a contradiction with \eqref{reduced en bounds pro}. 
\newline
We have thus proven that
\begin{equation}\label{pointwise g psi}
\left| |\gpm|^2 - g^2 \right| \leq g^2  \left| |u|^2 - 1 \right| \leq C \frac{|\log \ep|^{2}}{\ep^{3/4}}
\end{equation}
on $\At$. The result then follows by combining \eqref{pointwise bound hgvm} and \eqref{pointwise g psi}. 		
	\end{proof}

Theorem \ref{theo degree} follows as a corollary:

	\begin{proof}[Proof of Theorem \ref{theo degree}]
		\mbox{}	\\
		Given any $ \rb \leq r \leq 1 - \eps^{3/2} |\log\eps|^2 $, Theorem \ref{no vortices} guarantees that $ \deg\{u, \partial \B_r\} $ is well defined and independent of $ r $. Moreover one has
		\bdm
 			2\pi \lf| \deg\{u,\partial \B_r\} \ri| \leq \int_{\partial \B_r} \diff s \: |u|^{-1} \lf| \partial_{\tau} u \ri| \leq C \int_{\partial \B_r} \diff s \: \lf| \partial_{\tau} u \ri|,
		\edm
		because $u$ is bounded below in $\At$ as a consequence of the proof of Theorem \ref{no vortices}. Now integrating in $ r $ from $ \rb $ to $ 1 - \eps^{3/2} |\log\eps|^2 $ both sides of the above expression and using the fact that the degree is independent of $ r $ because $u$ has no vortices, we obtain
		\beq
 			2\pi \lf| \deg\{u,\partial \B_r\} \ri| \leq C \eps^{-1} |\log\eps|^{-1} \int_{\At} \diff \rv \: \lf| \nabla u \ri| \leq
			C \eps^{-1} |\log\eps|^{-1} \lf| \At \ri|^{1/2} \lf\| \nabla u \ri\|_{L^2(\At)},
		\eeq
		where we have used Cauchy-Schwarz inequality and the fact that $ |\At| = 2\pi (1 - \eps^{3/2} |\log\eps|^2 - \rb) = \OO(\eps|\log\eps|) $. On the other hand, \eqref{g lower bound} and \eqref{reduced en bounds pro} imply
		\[
		\lf\| \nabla u \ri\|_{L^2(\At)} \leq C \ep ^{1/2}|\log\eps|^{5/2}.
		 \] 
We conclude 
\[
2\pi \lf| \deg\{u,\partial \B_r\} \ri| \leq C |\log \ep| ^2
\]		 
and final result is thus a simple consequence of the definition \eqref{gvfunction u}.
	\end{proof}

\section{Rotational Symmetry Breaking}
\label{sec symm break proof}

We first introduce some notation that will be used in the proof of Theorem \ref{symmetry breaking}: The result stated there is equivalent to prove that no GP minimizer is a symmetric vortex, i.e., a wave function of the form $ f(r) \exp\{i n \vartheta \} $, $ n \in \Z $. We therefore denote by $ E_{n} $ the energy obtained by minimizing the GP functional on symmetric vortices, i.e.,
\beq
	E_{n} : = \inf_{f \in \gpdom} \gpf \lf[ f(r) \exp\{ i n \vartheta \} \ri] = \gpf\lf[ f_{n}(r) \exp\{ i n \vartheta \} \ri],
\eeq
where $ f_{n}(r)  $ is the unique real minimizer. 
\newline
We also define $ \bar n \in \N $ through $ \min_{n \in \Z} E_{n} = : E_{\bar n} $: Note that a minimizing $ \bar n $ certainly exists for any $ \eps $ thanks to the convexity in $ n $ of the functional. However such a minimizer needs not be unique because of some accidental degeneracy (there are at most 2 minimizers), which can be removed by a infinitesimal change of $ \eps $. 

The next lemma contains several useful properties of $ f_{\bar n} $:

	\begin{lem}[\textbf{Symmetric vortex minimizer}]
		\mbox{}	\\
		For any $ \eps > 0 $ and $ \Omega \gg \eps^{-1} $, there exists some $ \bar n \in \Z $ minimizing $ E_{\bar n} $ and it satisfies the estimate $ \bar n = \Omega (1 + \OO(\eps^{-1}\Omega^{-1})) $.
		\newline
		The associated minimizer $ f_{\bar n}(r) $ is unique and, up to multiplication by a constant phase factor, it is given by a positive radial function vanishing only at $ r = 0 $ and $r=1$. Moreover it has a unique maximum at some point $ 0 < \rmaxn < 1 $ and satisfies the $ L^2 $ estimate $	\lf\| f_{\bar{n}} \ri\|_{L^2(\ba\setminus\ba_{\rmaxn})} = o(1) $.
	\end{lem}

	\begin{proof}
		We first notice that by setting $ \bar n = : [\Omega] - \omega $ for some $ \omega \in \Z $, one can easily recover the coupled minimization problem studied in Proposition \ref{optimal phase ball pro} (see also Proposition \ref{optimal phase pro}) for some different angular velocity $ \Omega $. It is very easy to realize that the existence of a minimizing $ \omega $ (and thus $ \bar n $) as well as the estimate $ \omega = \OO(\eps^{-1}) $ can be deduced in the same way as in Proposition \ref{optimal phase ball pro}.
		
		On the other hand for any given $ n \in \Z $ the uniqueness and positivity of the minimizer $ f_n(r) $ can be deduced by standard arguments, whereas the existence of a unique maximum at some point $ 0 < \rmaxn < 1 $ can be proven by a rearrangement argument as in Proposition \ref{unique maximum} by noticing that the potential $ \bar{n}^2 r^{-2} $ is strictly decreasing.
	
		In order to prove the $L^2 $ estimate, we first notice that the fact that $ \bar{n} = \Omega (1 - \OO(\eps^{-1}\Omega^{-1}) $ implies the upper bound
		\beq
			\label{sym vortex en ub}
			E_{\bar n} \leq \tfe + \OO(\eps^{-2}) + \OO(\eps^{1/2}\Omega^{3/2}),
		\eeq
		which is a consequence of the pointwise estimates \eqref{exponential smallness} and \eqref{point est origin} together with the bound $ |B_{[\Omega] - \bar n}(r)| \leq \OO(\eps^{-1}) + \OO([\Omega]- \bar n) $ for any $ \rv \in \ann $. As in \eqref{chemical diff} the above estimate yields
		\beq
			\label{chemical diff sym vortex}
			\lf| \mu_{\bar n} - \tfchem \ri| \leq \OO(\eps^{-3/2}\Omega^{1/2}) + \OO(\eps^{-1/4}\Omega^{5/4}).
		\eeq
		We can thus repeat the proof of the pointwise estimate \eqref{pointwise bound 2} and the final result is
		\bdm
				\lf| \hgpm^2(r) - \tfm(r) \ri| \leq \OO(\eps^{1/2}\Omega^{1/2}) + \OO(\eps^{7/4} \Omega^{5/4}),
		\edm
		for any $ \rtf^2 + \eps^{-1} \Omega^{-1} |\log\eps|^{-1} \leq r^2 \leq 1 - \eps^{1/2} \Omega^{-1/2} |\log\eps|^{3/2} $. The argument described in Remark \ref{maximum pos rem} therefore gives	
		\beq
			\label{lb rmaxn}
			\rmaxn^2 \geq 1 - o(\eps^{-1}\Omega^{-1}),
		\eeq 
		which in addition to $ f_{\bar n}^2 \leq \OO(\eps\Omega) $ implies the result.
	\end{proof}

The main tool in the proof of the breaking of the rotational symmetry is the investigation of the second variation of the GP energy functional evaluated at some local minimizer: Given some $ \Psi $ solving the variational equation 
		\beq
			\label{GP variational local min}
			- \Delta \Psi - 2 \vec\Omega\cdot \vec L\, \Psi + 2 \eps^{-2} \lf| \Psi \ri|^2 \Psi = \mu_{\Psi} \Psi,
		\eeq
		where $ \mu_{\Psi} : = \gpf[\Psi] + \eps^{-2} \| \Psi \|_4^4 $, and some perturbation $ \Xi(\rv) \in H^1_0(\ba) $, one has $ \gpf[\Psi + \epsilon \Xi] = \gpf[\Psi] + \epsilon^2 \Q_{\Psi}[\Xi]  + \OO(\epsilon^3) $, where
		\beq
			\Q_{\Psi}[\Xi] : = \int_{\ba} \diff \rv \lf\{ \lf| \nabla \Xi \ri|^2  -2 \Xi^* \vec{\Omega} \cdot \vec{L} \Xi + 4 \eps^{-2} |\Psi|^2 |\Xi|^2 - \mu_{\Psi} |\Xi|^2 \right\} + 2 \eps^{-2} \Re \int_{\ba} \diff \rv \: (\Psi^*)^2 \Xi^2.
		\eeq
		By definition, if there exists some $ \Xi \in H^1_0 (\ba) $ such that $ \Q_{\Psi}[\Xi] < 0 $, the associated local minimizer $ \Psi $ is {\it globally unstable} and in particular can not be a global minimizer of the GP functional.

\begin{proof}[Proof of Theorem \ref{symmetry breaking}]
		\mbox{}	\\
		Assuming that the GP minimizer was given by a symmetric vortex $ f_{\bar n}(r) \exp\{i \bar n \vartheta\} $ for some $ \bar n $, we explicitly exhibit a trial function $ \Xi(\rv) $ such that the quadratic form $ \Q_{\Psi}[\Xi] $ evaluated at $ \Psi(\rv) =  f_{\bar n}(r) \exp\{i \bar n \vartheta\} $ is negative (for simplicity we denote it by $ \Q_{\bar n} $), which yields a contradiction with the assumption that the symmetric vortex is a global minimizer.
	
		For any $ d > 1 $ we set
		\beq
			\Xi(\rv) : = \lf(A(r) + B(r)\ri) e^{i(n+d)\vartheta} + \lf(A(r) - B(r) \ri) e^{i(n - d)\vartheta},
		\eeq
		with 
		\beq
			A(r) : = 
				\begin{cases}
					 r^{d +1} f_{\bar n}^{\prime}(r),	&	\mbox{if } 0 \leq r \leq \rmaxn,	\\
					0,						&	\mbox{if } \rmaxn \leq r \leq 1,	
				\end{cases}
				\qquad	
			 B(r) : = 
				\begin{cases}
					\bar n r^{d} f_{\bar n}(r),	&	\mbox{if } 0 \leq r \leq \rmaxn,	\\
					\bar n \rmaxn^{d} f_{\bar n}(r),						&	\mbox{if } \rmaxn \leq r \leq 1.	
				\end{cases}	
		\eeq
		A trial function of this form was first introduced in \cite[Theorem 2]{Seir} to prove symmetry breaking for a special class of trapping potential but here we replace in the original definition \cite[Eq. (2.30)]{Seir} $ d $ with $ - d $. Moreover in order to satisfy the condition $ \Xi \in H^1_0(\ba) $, we have modified the function $ A $ setting it equal to 0 for $ r \geq \rmaxn $. Note that the function certainly belongs to $ H^1(\ba) $ since $ f_{\bar n} $ is differentiable and $ A(r) + B(r) \sim r^{\bar n+d} $  as $ r \to 0 $, but $ B \notin H^2_0(\ba) $ because of the singularity in the derivative at $ r = \rmaxn $.
		\newline
		We can simply borrow the explicit computations from \cite[Eqs. (2.31) and (2.33)]{Seir} (recall that in our case there is no external potential, $ A $ vanishes for $ r \geq \rmaxn $, $ d $ has to be replaced with $ - d $ and $ \Omega $ with $ 2 \Omega $) and, denoting by $ \mu_{\bar n} $ the chemical potential associated with $ f_{\bar n} $, we obtain
		\beq
 			\label{evaluation Q}
			\Q_{\bar n}[\Xi] = 8 \pi \int_0^{\rmaxn} \diff r \: r^{2d + 2} f_{\bar n}(r) f_{\bar n}^{\prime}(r) \lf\{  (d +1) \mu_{\bar n} - \frac{2 (d + 1)}{\eps^2} f_{\bar n}^2 + 2 \Omega \bar n \ri\} + 4 \pi n^2 d^2 \int_{\rmaxn}^1 \diff r \: \frac{\rmaxn^{2d}}{r}  f^2_{\bar n}(r).
		\eeq
		
		Using the estimate \eqref{chemical diff sym vortex}, one immediately obtains that $ \chem = \mu_{\bar n} = - \Omega^2 (1 - o(1)) $. We can thus estimate the quantity between brackets in the first term in \eqref{evaluation Q} as
		\beq
 			(d +1) \mu_{\bar n} - 2 \eps^{-2} (d + 1) f_{\bar n}^2 + 2 \Omega \bar n \leq - \Omega^2 \lf(d -1 - o(1) \ri)
		\eeq
		which implies the bound
		\bml{
			\label{est Q 1}
			8 \pi \int_0^{\rmaxn} \diff r \: r^{2d + 2} f_{\bar n}(r) f_{\bar n}^{\prime}(r) \lf\{  (d +1) \mu_{\bar n} - 2 \eps^{-2} (d + 1) f_{\bar n}^2 + 2 \Omega \bar n \ri\} \leq	\\
			- 8 \pi \Omega^2 (d - 1 - o(1)) \int_0^{\rmaxn} \diff r \: r^{2d + 2} f_{\bar n}(r) f_{\bar n}^{\prime}(r) \leq	\\
			 - 4 \pi \Omega^2 (d - 1 - o(1)) \bigg[ \rmaxn^{2d+2} f_{\bar n}^2(\rmaxn) - (2d + 2) \rmaxn^{2d} \int_0^{\rmaxn} \diff r \: r f_{\bar n}^2(r) \bigg] \leq	\\
			- 4\pi \Omega^2 (d - 1 - o(1)) \rmaxn^{2d + 2} f^2_{\bar n}(\rmaxn) \lf( 1 - C d \eps^{-1} \Omega^{-1} \ri),
		}
		where we have used the fact that $ f $ is increasing between 0 and $ \rmaxn $ and the lower bound \eqref{lb rmaxn}.
		\newline
		On the other hand the last term in \eqref{evaluation Q} can be bounded as
		\beq
 			\label{est Q 2}
			 4 \pi n^2 d^2 \rmaxn^{2d} \int_{\rmaxn}^1 \diff r \: r^{-1}  f^2_{\bar n}(r) \leq 2 n^2 d^2  \rmaxn^{2d-2} \lf\| f_{\bar{n}} \ri\|^2_{L^2(\ba\setminus\ba_{\rmaxn})} \leq o(1) \rmaxn^{2d-2} \Omega^2 d^2.
		\eeq
		Hence \eqref{est Q 1} and \eqref{est Q 2} yield	
		\bdm
			\Q_{\bar n}[\Xi] \leq - 4 \pi d \Omega^2 \rmaxn^{2d+2}  f^2_{\bar n}(\rmaxn) \lf( 1 - d^{-1} - o(1) - C \eps^{-1} \Omega^{-1}d \ri) < 0
		\edm
		for any finite $ d \geq 2 $ and $ \eps $ small enough.
	\end{proof}

	\begin{rem}{\it (Flat Neumann case)}
		\mbox{}	\\
		The above proof applies with minor modifications to the case of the bounded trap $ \ba $ with Neumann conditions at the boundary $ \partial \ba $: It is indeed sufficient to make the replacements in the trial function $ \Xi $
		\bdm
			A(r) =  r^{d +1} f_{\bar n}^{\prime}(r),	\qquad	B(r) = \bar n r^{d} f_{\bar n}(r),
		\edm
		for any $ \rv \in \ba $ and compute
		\bdm
			\Q_{\bar n}[\Xi] = 8 \pi \int_0^{\rmaxn} \diff r \: r^{2d + 2} f_{\bar n}(r) f_{\bar n}^{\prime}(r) \lf\{  (d +1) \mu_{\bar n} - 2 (d + 1) \eps^{-2} f_{\bar n}^2 + 2 \Omega \bar n \ri\}.
		\edm
		Now since $ f_{\bar n} $ is increasing in the Neumann case, $ f_{\bar n}^2 \leq \OO(\eps\Omega) $ and $ \mu_{\bar n} = - \Omega^2(1 - o(1)) $, the quadratic form can be made negative for $ d $ large enough.
	\end{rem}

\section*{Appendix A}
\addcontentsline{toc}{section}{Appendix A}

\renewcommand{\theequation}{A.\arabic{equation}}
\setcounter{equation}{0}
\setcounter{subsection}{0}
\renewcommand{\thesection}{A}

In this Appendix we discuss some useful properties of the TF-like functionals involved in the analysis as well as the critical angular velocities.

\subsection{The Thomas-Fermi Functionals}

The minimization of the TF functional introduced in \eqref{TFf} has already been discussed in other papers (see, e.g., \cite[Appendix]{CY} or \cite[Appendix A]{CRY}), so we only sum up here the main results: The minimizer among positive functions is unique and explicitly given by 
\beq
	\label{TFm}
	\tfm(r) : = \half \lf[ \eps^2 \tfchem + \eps^2 \Omega^2 r^2 \ri]_+,
\eeq		
where $ [ \:\: \cdot \:\: ]_+ $ stands for the positive part and $ \tfchem : = \tfe + \eps^{-2} \lf\| \tfm \ri\|^2_2 $. If $ \Omega \geq 2(\sqrt{\pi} \eps)^{-1}$, the chemical potential is given by $ \tfchem = - \Omega^2 \rtf^2 $ with 
\beq
	\label{rtf}
	\rtf : = \sqrt{1 - \frac{2}{\sqrt{\pi} \eps \Omega}},
\eeq
and the TF minimizer can be rewritten as $ \tfm(r) = \half \eps^2 \Omega^2 \lf[ r^2 - \rtf^2 \ri]_+ $, which makes explicit the fact that it vanishes for $ r \leq \rtf $.

The corresponding ground state energy can be explicitly evaluated and is given by
\beq
	\label{TFe explicit}
	\tfe =
		\begin{cases}
			\frac{1}{\pi} \eps^{-2} - \half \Omega^2 - \frac{1}{48} \pi \eps^{2} \Omega^4,      	&   \mbox{if } \Omega \leq \Osec,    \\ 
            	- \Omega^2 \lf[ 1 - 4/(3\sqrt{\pi}) \Omega \ri],   						&   \mbox{if } \Omega > \Osec.
		\end{cases}
\eeq
Note that above the second critical velocity, the annulus $ \tfd : = \tfsupp $ has a shrinking width of order $ \eps |\log\eps| $ (see \eqref{rtf}) and the leading order term in the ground state energy asymptotics is $ - \Omega^2 $, which is due to the convergence of $ \tfm $ to a distribution supported at the boundary of the trap.

In the giant vortex regime another TF-like functional becomes more relevant, i.e., 
\beq
 	\label{hTFf}
	\htff_{\omega}[\rho] : = \int_{\B} \diff \rv \: \lf\{ - \Omega^2 r^2 \rho + B_{\omega}^2(r) \rho + \eps^{-2} \rho^2 \ri\} = \int_{\B} \diff \rv \: \lf\{ \lf( [\Omega] - \omega \ri)^2 r^{-2} \rho + \eps^{-2} \rho^2 \ri\} - 2 \Omega [\Omega - \omega],
\eeq
where the potential $ \rmagnp $ is defined in \eqref{rmagnp}, $ \omega \in \Z $ and we have used the normalization in $ L^1(\B) $ of the density in the last term. The minimization of such a functional was studied in details in \cite[Appendix A]{CRY} and we recall here only the most important fact, i.e., the ground state energy $ \htfe $ satisfies the estimate
\beq
	\label{hTFe}
	\htfe_{\omega} = \tfe + \bigg[ \omega - \frac{2}{3\sqrt{\pi}\eps} \bigg]^2 + \frac{2}{9 \pi \eps^2} + \OO(\eps^{-2} |\log\eps|^{-1}),
\eeq
which suggests that it is minimized by a phase $ \optphtf : = 2(3\sqrt{\pi}\eps)^{-1} $.

\subsection{The Third Critical Angular Velocity $ \Othird $}
\label{ang vel 3 sec}

In this last part of the Appendix we state the estimate of the critical velocity $ \Othird $, which is defined as the angular velocity at which vortices disappear from the bulk of the condensate. To estimate this velocity we need to compare the vortex energy cost $ \frac{1}{2} \hgvmo^2(r) |\log\eps| $ with the vortex energy gain $ |\fin(r)| $ (see \eqref{Fin}, \eqref{F} and \eqref{Fout}). In \cite[Appendix]{CRY} a similar comparison is performed when the density $ \hgvmo^2 $ is replaced by $ \tfm $ and it is shown that, if $ \Omega > \Othird $ in the sense that $ \Omega = \Omega_0 \eps^{-2} |\log\eps|^{-1} $ with $ \Omega > 2(3 \pi)^{-1} $, then the function
\beq
	\label{gainTF}
	\gaintf(r) : =\half |\log\eps| \tfm(r) - \lf|\costtf(r) \ri|,
\eeq
where
\beq
	\label{exp TF potential}
	\costtf(r) : = 2 \int_{\rtf}^r \diff s \: \vec{B}_{\optphtf}(r) \cdot \vec{e}_{\vartheta} \: \tfm(r),
\eeq
satisfies the lower bound
\beq
	\label{gainTF lb}
	\gaintf(r) \geq  C \eps^{-1} |\log \ep| ^{-2} > 0 
\eeq
for any $\vec{r}$ such that $ r \geq \rb  = \rtf + \eps|\log\eps|^{-1} $.

The analogous result for the original function
\beq
	\gain(r) : = \half \hgvmo^2(r) |\log\eps| - \lf|\fin(r)\ri|,
\eeq
is proven in the following

\begin{pro}[\textbf{Third critical velocity $ \Othird $}] 
		\label{ang vel 3 pro}
		\mbox{}	\\
		If $ \Omega_0 > 2(3\pi)^{-1} $ and $ \eps $ is small enough, there exists a finite constant $ C $ such that 
		\[ 
			\gain(r) \geq C \eps^{-1} |\log \ep| ^{-2} > 0 
		\] 
		for any $\vec{r}$ such that $ r \geq \rb = \rtf +  \ep |\log \ep| ^{-1} $.
	\end{pro}

	\begin{proof}
		The result can be proven in the same way as \cite[Proposition A.2]{CRY} by noticing that $ |\fout| \leq |F(1)| \leq \OO(1) $ and using such an estimate to replace $ \fin(r) $ with $ F(r) $ in $ \gain(r) $.
	\end{proof}

\vspace{1cm}
\noindent{\bf Acknowledgements.} The work of NR was supported partly by {\it R\'{e}gion Ile-de-France} through a PhD grant and partly by the European 
Research Council under the European Community's Seventh Framework Programme 
(FP7/2007--2013 Grant Agreement MNIQS no. 258023). Part of this work was realised when NR was visiting the university of Aarhus thanks to the hospitality of S\o{}ren Fournais.

\end{document}